\documentclass[a4paper, 11 pt]{article}
\usepackage [a4paper,left=3.1 cm,bottom=2.7cm,right=3.1 cm,top=2.7cm]{geometry}
\usepackage[english]{babel}
\usepackage{amssymb}
\usepackage{amsmath,amsthm}
\usepackage{csquotes}
\usepackage{subcaption,hyperref}
\usepackage{tabularx,array}
\usepackage{graphicx}
\usepackage{enumitem} 
\usepackage{mathtools}
\usepackage{xcolor}
\def\E{\mathbb{E}}
\def\P{\mathbb{P}}

\newtheorem{theorem}{Theorem}
\newtheorem{lemma}{Lemma}

\newtheorem{proposition}{Proposition}

\newtheorem{h}{Hypothesis}

\newcommand{\keywords}[1]{\par\addvspace\baselineskip
\noindent\enspace\ignorespaces#1}

\begin{document}

\title{On the implied volatility of Inverse options under stochastic volatility models}

\author{Elisa Al\`os$^{\ast}$, Eulalia Nualart$^{\ast}$ and Makar Pravosud\thanks{ Universitat Pompeu Fabra and Barcelona School of Economics, Department of Economics and Business, Ram\'on Trias Fargas 25-27, 08005, Barcelona, Spain. 
EN acknowledges support from the Spanish MINECO grant PID2022-138268NB-100 and
Ayudas Fundacion BBVA a Equipos de Investigaci\'on Cient\'ifica 2021.}}
\maketitle

\begin{abstract}
In this paper, we study the short-time behavior of at-the-money implied volatility for inverse European options with a fixed strike price. The asset price is assumed to follow a general stochastic volatility process. Using techniques from Malliavin calculus, such as the anticipating It\^o's formula, we first compute the implied volatility of the option as its maturity approaches zero. Next, we derive a short-maturity asymptotic formula for the skew of the implied volatility, which depends on the roughness of the volatility model. We also demonstrate that our results can be easily extended to Quanto-Inverse options. We apply our general findings to the SABR and fractional Bergomi models and provide numerical simulations that confirm the accuracy of the asymptotic formula for the skew. Finally, we present an empirical application using Bitcoin options traded on Deribit, showing how our theoretical formulas can be applied to model real market data for such options.
\end{abstract}


\section{Introduction}

Over the last several decades, option pricing models have been developed for conventional assets such as stocks, bonds, interest rates, and foreign currencies. Today, cryptocurrency derivatives, financial contracts whose value depends on an underlying cryptocurrency asset, have emerged as a new class of securities that have garnered significant attention. A key peculiarity of crypto derivatives is how one defines a cryptocurrency: Is it a security, a currency, or a commodity? With regard to options, the answer to this question affects the pricing methodology. A detailed discussion on this topic can be found in Alexander {\it{et al.}} \cite{Alexander2021}. Unfortunately, there is no clear legal answer to this question; see Bolotaeva {\it{et al.}} \cite{Bolotaeva2019}. However, a closer examination of the issue leads to the following conclusions.

Cryptocurrency (at least Bitcoin and Ethereum) cannot be considered a security, as it is fully decentralized, and no central authority controls its issuance, whereas securities are issued by a central authority. Furthermore, cryptocurrencies cannot be treated as conventional (fiat) currencies. A key question is whether they preserve the fundamental characteristics of money. Although cryptocurrencies can occasionally be used to buy and sell goods, they are not widely accepted as a means of payment. Furthermore, historical data shows the extreme volatility of cryptocurrencies, leading to the conclusion that their purchasing power is not stable enough over time, and therefore they cannot serve as a store of value. In contrast, Central Bank Digital Currency addresses the volatility issue by controlling the issuance. Finally, although some companies may accept cryptocurrencies as payment, most still use traditional currencies to measure the value of goods and services. For further discussion on the similarities between cryptocurrencies and regular currencies, see Hazlett and Luther \cite{Hazlett2020} and Ammous \cite{Ammous2018}.

On the other hand, if we consider the classical Garman and Kohlhagen \cite{Garman} foreign exchange (FX) pricing model, the construction of the delta-hedged portfolio for FX options is conceptually different from regular options, as we cannot buy and sell units of the FX spot rate. Instead, hedging is done by buying and selling units of the underlying foreign bond. This difference fundamentally undermines the idea of pricing crypto options using FX models, as cryptocurrencies can be bought and sold in a manner similar to regular tradable assets.

Alternatively, some people consider Bitcoin to be digital gold, but how similar is it to the actual commodity? In Goutte{\it{et al.}} \cite{cryptofinance}, the authors define the following characteristics of hard commodities: they are costly to mine or extract, storable, not controlled by any single government or institution in terms of global supply, demand, or price, and they have intrinsic value, which means that they can be consumed or used as inputs in the production of other goods. The first three properties are naturally satisfied by Bitcoin, but the fourth remains debatable. As a result, we cannot conclusively classify crypto as a commodity. For a more detailed discussion, see Ankenbrand and Bieri \cite{Ankenbrand2018} and Gronwald \cite{Gronwald2019}. Finally, Dyhrberg \cite{DYHRBERG201685} provides an empirical analysis using GARCH models for Bitcoin volatility to explore whether Bitcoin is more similar to gold or the US dollar.

The aim of this paper is to solve the pricing problem of crypto options using only the payoff function. A natural way to define the payoff of crypto options is through Inverse options, which are settled in cryptocurrency rather than fiat currency. We will focus on the case of an Inverse European call, whose payoff is given by: $$\left(\frac{S_T-K}{S_T}\right)_+,$$ 
where $(x)_+=\max(x,0)$,
$S_T$ denotes the price of Bitcoin in US dollars at maturity $T$ and $K$ is a fixed strike.
In simple terms, if the option becomes in-the-money, the payoff is made in cryptocurrency rather than fiat currency.

Inverse European options are the only type of options traded on the Deribit exchange, which controls over 80 percent of the global crypto options market. For example, on June 10, 2023, the open interest in Bitcoin options on Deribit was 7.5 billion, while the closest competitors, OKX and Binance, had open interest of 0.5 billion and 0.17 billion, respectively. As a result, accurately pricing and hedging Inverse European options is crucial from a practical standpoint. However, this is quite challenging due to the mechanics of the Deribit exchange.
Deribit does not allow fiat currency, and all options are margined in cryptocurrency. This is particularly beneficial for professional crypto traders. For instance, consider a crypto hedge fund or a crypto market maker, businesses that deal exclusively with crypto assets. It is natural for them to manage their trading books in cryptocurrency rather than fiat currency. While they are exposed to cryptocurrency depreciation risk, it is easier to manage this risk at the book level rather than on a trade-by-trade basis. This is one of the key reasons for the development of Inverse European options.

The literature on crypto derivatives is still relatively new but has been gaining increasing attention from researchers. For example, Alexander {\it{et al.}} \cite{Alexander2021} price Inverse European options under the constant volatility Black-Scholes model. Matic {\it{et al.}} \cite{Matic2023} study empirical hedging of Inverse options under different stochastic volatility models. Alexander {\it{et al.}} \cite{Alexander2022} use a delta adjusted for skew to hedge Inverse options under constant and local volatility. Hou {\it{et al.}} \cite{Hou2020} price Inverse options under stochastic volatility models with correlated jumps. Finally, Siu and Elliott \cite{Siu2021} use the SETAR-GARCH model to model Bitcoin return dynamics.

In this paper, we study the behavior of implied volatility for Inverse European options under general stochastic volatility models. Specifically, we provide general sufficient conditions on a stochastic volatility model to derive short-time maturity asymptotic formulas for the at-the-money level and the skew of the implied volatility of Inverse options. We then apply these results to two well-known models: the SABR and fractional Bergomi models.

The main tool for proving these results is the anticipating It\^o's formula from Malliavin calculus; see Appendix \ref{MCintro} for an introduction to this topic. The first step is to apply this formula using a similar approach to Alòs \cite{Alos2006a}, in order to derive a decomposition formula for the option price; see Theorem \ref{decomposition}. This approach is also used in Alòs {\it{et al.}} \cite{Alos2007b} for jump-diffusions with stochastic volatility. In that paper, the short-time limit of the implied volatility skew for European calls is also derived, which corresponds to the second step of our paper. This approach has been further developed in other contexts. Recently, Alòs {\it{et al.}} \cite{Alos2022a} applied this methodology to VIX options, and Alòs {\it{et al.}} \cite{Alos2022} extended it to Asian options. These two papers demonstrate how studying Asian or VIX options under stochastic volatility reduces to studying European-type options, where the underlying is represented by a stochastic volatility model with a modified volatility process that depends on maturity.

In the present paper, we extend the methodology outlined above to Inverse European calls. The main novelty, compared to the papers mentioned, is the following.
\begin{itemize}
\item To the best of our knowledge, we are the first to provide a rigorous analytical study of the asymptotic behavior of the implied volatility of Inverse European options under fractional and stochastic volatility models. When dealing with Inverse options, the Black-Scholes formula differs from the classical one, so our results cannot be directly derived from existing literature. Instead, we work with the Black-Scholes formula for Inverse European calls under constant volatility, as obtained in Alexander 
{\it{et al.}} \cite{Alexander2021}. This formula turns out to be a non-monotonic function (see Figure \ref{fig9}), meaning its inverse is not uniquely defined, which complicates computations since implied volatility is defined as the inverse of the Black-Scholes formula. We overcome this challenge by noting that for sufficiently small maturities, which is the case studied in this paper, the function is strictly monotonic. However, the inverse is not explicit as in the classical Black-Scholes case, so a careful study of its behavior and derivatives is required, as shown in Appendix B.

\item 
Our formulas can be applied to general stochastic and fractional volatility models. We provide numerical simulations using the SABR and fractional Bergomi models. Additionally, we present an empirical example with Bitcoin options to demonstrate how the results of this paper can be used to model real market data for such options, which have become increasingly popular in financial markets, as discussed at the beginning of this introduction.
\end{itemize}

Understanding the behavior of the implied volatility of Inverse options is crucial for hedging purposes and can also be used to derive approximation formulas for their price, as mentioned at the end of Section 2. In a world with a flat implied volatility surface, one could simply use the Black-Scholes delta to replicate the derivative payoff. However, in the presence of implied volatility skew, the classical Black-Scholes hedge will not be accurate. The correct approach is to use the Black-Scholes delta with a correction based on the implied volatility skew. In our empirical example, we focus on the fractional Bergomi model. 

One of the key parameters in this model is the Hurst exponent \( H \). Its value determines the skew that can be reproduced from market data. Since different markets exhibit different behaviors, we expect to see varying values of \( H \) that best fit the observed data. For example, in equity markets, the implied volatility skew exhibits a blow-up, which can be reproduced with \( H < 0.5 \), corresponding to rough volatility See, for instance, Cont and Das \cite{CD24}. In our case, we find that \( H > 0.5 \) best fits our market data and produces zero skew.

Additionally, our analysis allows us to derive the same formulas for Quanto-Inverse European options, whose payoff is given by:
$$
\left(\frac{R(S_T-K)}{S_T}\right)_+,$$
where \( R \) is a fixed exchange rate. From a mathematical perspective, Quanto-Inverse options do not differ significantly from Inverse options. In particular, as shown in Section 2, the implied volatility level and skew of both options are equal up to a factor of \( R \), so it is sufficient to study Inverse options. However, from a practical perspective, Inverse and Quanto-Inverse options differ, as they are traded on different markets. Unlike Inverse options, which are available to the general public, Quanto-Inverse options are over-the-counter derivatives.

The paper is organized as follows: Section 2 presents the problem statement and the main result, Theorem 1, where the at-the-money implied volatility level is derived in (\ref{main1}) and the skew in (\ref{main2}). The intermediate steps leading to the proof of this theorem are provided in Section 3. Specifically, Theorem 2 enables us to obtain (\ref{main1}), and Propositions 1 and 2 are crucial for deriving (\ref{main2}). Note that Proposition 1 relies on (\ref{main1}). The proof of Theorem 1 is provided in Section 4. Section 5 presents a numerical study for the SABR and fractional Bergomi models, along with an empirical application to the implied volatility of Bitcoin options.

\section{Statement of the problem and main results}\label{statement-of-the-problem-and-notation}

Consider the following model for the asset price $S_t$ on the time interval $[0,T]$
\begin{equation} \begin{split}\label{B_Epm}
dS_t &= \sigma_tS_tdW_t\\
 W_t &= \rho W_t' + \sqrt{(1-\rho^2)}B_t,
 \end{split}
\end{equation} where $S_0>0$ is fixed and $W_t$, $W_t'$, and $B_t$ are three standard Brownian motions on $[0,T]$
defined on the same risk-neutral complete probability space $(\Omega, \mathcal{G}, \mathbb{P})$. We denote by $\E$ the expectation with respect to $\P$. For simplicity,  we assume that the interest rate is zero, see for instance in Al\`os {\it{et al.}} \cite{Alos2022a}. We assume that \(W_t'\) and
\(B_t\) are independent and \(\rho\in [-1,1]\) is the correlation coefficient between \(W_t\) and
\(W_t'\). 

We consider the following hypotheses on the volatility process $\sigma_t$:
\begin{h} \label{Hyp1}
The process $\sigma=(\sigma_t)_{t \in [0,T]}$ is square integrable, adapted to the filtration generated by $W'$, a.s. positive and continuous, and satisfies that for  all $t \in [0,T]$,
$$
c_1 \leq \sigma_t \leq c_2,
$$
for some positive constants $c_1$ and $c_2$.
\end{h}

\begin{h}
For all $p\geq 1$ there exist $c>0$ and $\gamma >0$ such that for all $0\leq s\leq r\leq T \leq 1$,
$$(\mathbb E|\sigma_r-\sigma_s|^p)^{1/p} \leq  c (r-s)^{\gamma}.$$
\end{h}

\begin{h}
For all $p \geq 2$, $\sigma \in \mathbb{L}^{2,p}_{W'}$ (see Appendix \ref{MCintro} for the definition of this space).
\end{h}

\begin{h}\label{Hyp4}
There exists  $H\in (0,1)$ and for all $p \geq 1$ there exist constants $c_1,c_2>0$ such  that for $0 \leq s \leq r \leq t \leq T \leq 1$ a.e.,
\begin{equation} \label{d1}
\{\E (\vert D_r^{W'}\sigma_t \vert^p)\}^{1/p}\leq c_1 (t-r)^{H-\frac{1}{2}}
\end{equation}
and
\begin{equation} \label{d2}
\{\E (\vert D_s^{W'} D_r^{W'}\sigma_t \vert^p)\}^{1/p}\leq c_2 (t-r)^{H-\frac{1}{2}} (t-s)^{H-\frac{1}{2}},
\end{equation}
where $D^{W'}$ denotes the Malliavin derivative defined in Appendix \ref{MCintro}.
\end{h}

We observe that when the volatility $\sigma_t$ is constant, the model (\ref{B_Epm}) is the usual Black-Scholes model. We also observe that although the volatility can be driven by a fractional Brownian motion (see Section 5.2), we are assuming that the volatility is not a traded asset. Thus, there exists a risk-neutral probability $\P$ and there are no arbitrage opportunities (see Comte and Renault \cite{CR98}). Moreover, if  $(V_t)_{t \in [0,T]}$ and $(V_t^{Q})_{t \in [0,T]}$ denote the values of an Inverse European call and a Quanto Inverse European call options with fixed strike $K$, respectively, we have that
$$V_0=\mathbb E\left(\frac{S_T-K}{S_T}\right)_{+}=K\times\mathbb E(K^{-1} - S^{-1}_T)_{+},$$  
and
$$V_0^{Q}=\mathbb E \left(\frac{R(S_T-K)}{S_T}\right)_{+}=K \times R \times \mathbb E(K^{-1} - S_T^{-1})_{+},$$
where $R$ is a fixed exchange rate. Recall that $\E$ denotes the expectation under a risk-neutral probability $\P$. Note that as in \cite{Alos2007b}, we assume that the market selects
a unique risk-neutral measure $\P$ under which these derivative contracts are priced.

Notice that the difference between \( V_0 \) and \( V_0^Q \) arises from the currency in which the options are quoted. In our case, \( V_0 \) represents the crypto price of the option, while \( V_0^Q \) represents the dollar value of the option.

We denote by \(BS(t,x,k,\sigma)\) the Black-Scholes price
of an Inverse European call option with time to maturity \(T-t\), log-underlying price \(x\), log-strike price \(k\) and
volatility \(\sigma>0\). Then, it is well-known that (see Alexander {\it{et al.}} \cite{Alexander2021}),
\begin{align*}
\begin{split}
  BS(t,x,k,\sigma)&=N(d_{2}(k,\sigma))-e^{\sigma^2(T-t)}e^{k-x}N(d_{1}(k,\sigma)),\\
  d_{2}(k,\sigma)&=\frac{x-k}{\sigma\sqrt{T-t}}-\frac{\sigma}{2}\sqrt{T-t},\\
  d_{1}(k,\sigma)&=d_{2}(k,\sigma) - \sigma\sqrt{T-t},
\end{split}  
\end{align*}  
 where \(N\) is the cumulative distribution function of a standard normal random variable. Moreover, the Black-Scholes price for a Quanto Inverse European call is given by $BS^{Q}(t,x,k,\sigma) = R \times BS(t,x,k,\sigma)$. 
 
 One can easily check that the Black-Scholes price satisfies the
following PDE 
\begin{align}
  \partial_t BS(t,x,k,\sigma)-\frac{1}{2}\sigma^2 \partial_x BS(t,x,k,\sigma) + \frac{1}{2}\sigma^2 \partial^2_{xx} BS(t,x,k,\sigma) = 0.
\label{bspde}
\end{align}
Moreover, one can also check that the classical relationship between the Gamma, the Vega and the Delta holds, that is,
\begin{equation} \label{gvd}
\frac{\partial_{\sigma}BS(t,x,k,\sigma)}{\sigma(T-t)}=\partial_{xx}^2BS(t,x,k,\sigma)-\partial_xBS(t,x,k,\sigma).
\end{equation} 

As in Al\`os \cite{Alos2006a}, we consider the log price $X_t=\log(S_t)$, which satisfies
\begin{equation*} 
dX_t=\sigma_t dW_t-\frac12 \sigma_t^2 dt.
\end{equation*}

Next, we observe that, as $BS(T,x,k,\sigma)=e^k\times(e^{-k}-e^{-x})_+$,  the price of an Inverse call option $V_0=e^k\times\E(e^{-k}-e^{-X_T})_+$ can be written as 
\begin{equation*} 
V_0=\mathbb{E}(BS(T,X_T,k,v_T)), \qquad \text{where} \quad v_t=\sqrt{\frac{1}{T-t}\int_t^T \sigma_s^2ds} \quad \text{ if } t<T,
\end{equation*}
and by continuity we have $\lim_{t\to T} v_t = \sigma_T$. 
In particular, \(V_T=BS(T,X_T,k,v_T)\). Notice that $V^{Q}_t=R \times V_t$. This implies that the implied volatility level and skew of Inverse and Quanto Inverse European call options are equal up to the factor $R$. Hence, we will only state the main results of this paper for the Inverse options. 

We define the at-the-money implied volatility (ATMIV)  of an Inverse European call option as the quantity $I(0,k^{\ast})$ satisfying $$V_0=BS(0,X_0,k^{\ast},I(0,k^{\ast})),$$ where $k^{\ast}=X_0$. That is,
\(I(0,k^{\ast})=BS^{-1}(0,X_0,k^{\ast},V_0)\).
As it is shown in Appendix B, this inverse is only well-defined for $T$ sufficiently small, which is the case we are studying in this paper.
Precisely, the aim of this paper is to apply  the Malliavin calculus techniques developed in Al\`os \cite{Alos2006a} in order to obtain formulas for  $$\lim_{T\to 0}I(0,k^{*}) \quad \text{and} \quad 
\lim_{T\to 0}\partial_kI(0,k^{*})$$ 
under the general stochastic volatility model (\ref{B_Epm}).

The main results of this paper are given in the following theorem.
\begin{theorem}  \label{limskew}
Assume Hypotheses 1-4. Then, 
\begin{align} \label{main1}
\lim_{T\to 0}I(0,k^*)=\sigma_0.
\end{align}
Moreover,
\begin{align} \label{main2}
\begin{split}
\lim_{T \to 0} T^{\max(\frac12-H, 0)}\partial_kI(0,k^{*})&=\lim_{T \to 0} T^{\max(\frac12-H, 0)}\frac{\rho}{ \sigma_0 T^2}\int_0^T\left(\int_r^T\E(D_r^{W'}\sigma_u) du \right) dr,
 \end{split}
\end{align}
provided that both limits exist.
\end{theorem}

We observe that when prices and volatilities are uncorrelated  then the short-time skew equals to zero. Observe also that since the term $\E(D_r^{W'}\sigma_u)$ is of order $(u-r)^{H-\frac12}$ (see Hypothesis 4), the quantity multiplying the term $T^{\max(\frac12-H, 0)}$ of the right hand side of (\ref{main2}) is bounded by $cT^{H -\frac12}$. In particular, its  limit  is $0$ if $H>1/2$. This suggests that, in the case $H<\frac12$, we need to multiply by $T^{\frac12-H}$ in order to obtain a finite limit. See the examples in Section 5.

The results of Theorem  \ref{limskew} can be used in order to derive approximation formulas for the price of an Inverse and Quanto Inverse European call options. Notice that, as
\begin{equation*} 
V_0= BS(0,X_0,k,I(0,k)).
\end{equation*}
by Taylor's formula we can use the approximations
\begin{equation*} 
\begin{split}
I(0,k) &\approx I(0,k^{*}) + \partial_kI(0,k^{*})(k-k^{*}).
\end{split}
\end{equation*}
Of course, this approximation is only linear, and one would expect to obtain better results if one has a short maturity asymptotic formula for the curvature $\partial_{kk}^2I(0,k^{*})$. The short-time maturity asymptotics for the at-the-money curvature of the implied volatility for European calls under general stochastic volatility models is computed in Al\`os and Le\'on \cite{Alos2017}. A Taylor expansion for short maturity asymptotics for Asian options when the underlying asset follows a local volatility model is obtained in Pirjol and Zhu \cite{Pirjol2016}.
In our setting, computing the curvature is more challenging and we leave it for further work.

\section{Preliminary results}

In this section, we provide closed-form decomposition formulas for the price and the ATMIV skew of an Inverse call option under the stochastic volatility model (\ref{B_Epm}).

We begin with the following preliminary lemma. For the standard European call option case, see Lemma 6.3.1 in Alòs {\it{et al.}} \cite{Alos2022a}.
\begin{lemma} \label{bound}
Assume Hypothesis 1.  Then, for all $p \geq 1$ there exist positive constants $C_1(p)$, $C_2$ and $C_3$ such that for all $0\leq s <T<1$,
\begin{align}
\label{o1}
\left(\E(\left|H(s,X_s,k,v_{s})\right|^p)\right)^{1/p}& \leq  C_1(p) \left( T-s \right)^{-1}\\ \label{o2}
\left|\partial_{x}G(s,X_s,k,v_{s})\right| & \leq  C_2 \left ( T-s \right)^{-2}, \\ \label{o3}
\left|(\partial_{xxx}^3-\partial_{xx}^2)G(s,X_s,k,v_{s})\right| & \leq  C_3 \left ( T-s \right)^{-3},
\end{align}
where $H(s,x,k,v_s)=\frac{1}{2}(\partial_{xxx}^3BS(s,x,k,v_s)-\partial_{xx}^2BS(s,x,k,v_s))$ and \\$G(s,x,k,v_s)=\partial_k H(s,x,k,v_s) - H(s,x,k,v_s)$.
\end{lemma}

\begin{proof} We start proving (\ref{o1}). Straightforward differentiation gives us the following 
\begin{align*}
\begin{split}
H(s,x,k,v_s) &=-\frac{\left(2 k-3 v_s ^2 (T-s)-2 x\right) \exp \left(-\frac{\left(2 k+v_s ^2 (T-s)-2 x\right)^2}{8 v_s ^2 (T-s)}\right)}{2 \sqrt{\pi }v_s ^3 (T-s)^{3/2}}\\
  &\qquad  -\frac{\sqrt{\pi }v_s ^3 (T-s)^{3/2}e^{k+v_s ^2 (T-s)-x}\text{Erfc}\left(\frac{2 k+3 v_s ^2 (T-s)-2 x}{2 \sqrt{2}v_s 
   \sqrt{T-s}}\right)}{2 \sqrt{\pi }v_s ^3 (T-s)^{3/2}}
\end{split}
\end{align*}
where $\text{Erfc}(z)=\frac{2}{\sqrt{\pi}}\int_z^{\infty} e^{-t^2}dt$ if $z \geq 0$ and $\text{Erfc}(z)=\frac{2}{\sqrt{\pi}}\int_{-\infty}^{z} e^{-t^2}dt$ if $z < 0$. 

Using Hypothesis 1, the fact that for all $a>0$ and $b>0$ the function $z^a e^{-b z^2}$ is bounded, and that $T<1$ it is easy to see that the first term is bounded by $c (T-s)^{-1}$.

For the second term, we use the fact that the function Erfc is bounded and that $e^{\sup_{s \in [0,T]} \vert X_s \vert}$ has bounded moments of all orders by Hypothesis 1. This completes the proof of (\ref{o1}).

We next prove (\ref{o2}).
Straightforward differentiation gives
\begin{equation*}
\begin{split}
\partial_{x}G(s,x,k,v_{s}) &= \exp \left(-\frac{\left(2(k-x)+(T-s) v_s^2\right){}^2}{8 (T-s) v_s^2}\right)   \\
   &\times \bigg(\frac{8 (k-x)^3-2 (k-x+2) (T-s)^2 v_s^4}{16 \sqrt{2 \pi } (T-s)^{7/2} v_s^7} \\
   &+ \frac{4 (k-x-6) (k-x) (T-s) v_s^2-(T-s)^3 v_s^6}{16 \sqrt{2 \pi } (T-s)^{7/2} v_s^7}\bigg).
\end{split}
\end{equation*}
Then, due to Hypothesis \ref{Hyp1}, we get that
\begin{equation*}
\begin{split}
&\left|\partial_{x}G(s,x,k,v_{s})\right|  \leq \frac{1}{(T-s)^{2} v_s^4}\exp \left(-c_0\left(y+v_s\sqrt{T-s}\right)^2\right)  \\
& \quad \qquad \times \Biggl( c_1\left|y\right|^3 + 
c_2\left|y \right|^2 + c_3\left| y\right|+c_4(T-s)^{3/2} v_s^3 \Biggr) \\
&\quad  \leq \frac{1}{(T-s)^{2} v_s^4}\exp \left(-c_0\left(y+v_s\sqrt{T-s}\right)^2\right)  \\ 
& \quad \qquad \times \Biggl( c_1'\left(\left|y+v_s\sqrt{T-s}\right|^3 +\left(v_s\sqrt{T-s}\right)^3\right) 		\\
&\qquad \qquad + 
c_2'\left(\left|y +v_s\sqrt{T-s}\right|^2  +\left(v_s\sqrt{T-s}\right)^2\right)\\
& \qquad\qquad  + c_3'\left(\left| y+v_s\sqrt{T-s}\right|+v_s\sqrt{T-s}\right)+c_4'\Biggr),
\end{split}
\end{equation*}
where $y=\frac{\left(k-x\right)}{\sqrt{(T-s)}v_s}$.

Hence, using the fact that for all $a \geq 0$ and $b>0$ the function $z^a e^{-b z^2}$ is bounded and $T<1$, we conclude that
(\ref{o2}) holds true.

Finally, we have that
\begin{equation*}
\begin{split}
&\left|(\partial_{xxx}^3-\partial_{xx}^2)G(s,x,k,v_{s})\right| \leq \frac{1}{(T-s)^{3} v_s^6}\exp \left(-c_0 \left(y+v_s\sqrt{T-s}\right)^2\right)  \\
& \qquad \times  \Biggl( c_1 |y|^5 + c_2|y|^4 + c_3 |y|^3 +c_4 |y|^2 +c_5 |y| +c_6 (T-s)^{7/2} v_s^5 \Biggr),
\end{split}
\end{equation*}
and the same argument as above allows us to complete the proof (\ref{o3}), and thus the proof of the Lemma.
\end{proof}

The first result of this section is the decomposition of the price of inverse options. The proof follows exactly the same lines as Theorem 4.2 in Alòs {\it{et al.}} \cite{Alos2007b}. See also Theorem 6 in Alòs \cite{Alos2006a} and Theorem 26 in Alòs 
{\it{et al.}} \cite{Alos2022a}. Note that the derivatives of the Black-Scholes formula for inverse options satisfy the same equations, (\ref{bspde}) and (\ref{gvd}), which is the key ingredient of the proof, together with the anticipating Itô's formula. For completeness, we provide a sketch of the proof, as the derivatives of the Black-Scholes formula are different. Thus, Lemma \ref{bound} (instead of Lemma 2 in \cite{Alos2007b}) is needed to verify that all the integrals are well-defined.
\begin{theorem} \label{decomposition}
Assume Hypotheses 1-4. Then, the following relation holds 
\begin{equation*}
\begin{split}
V_0= \mathbb{E}\left(BS(0,X_0,k,v_0)\right)+\mathbb{E}\left(\int_0^T H(s,X_s,k,v_s)\sigma_s \left(\int_s^T D_s^W\sigma_r^2dr\right) ds\right),
\end{split}
\end{equation*}
where recall that $H(s,x,k,v_s)=\frac{1}{2}(\partial_{xxx}^3BS(s,x,k,v_s)-\partial_{xx}^2BS(s,x,k,v_s))$.
\end{theorem}

\begin{proof}  
Since \(V_T=BS(T,X_T,k,v_T)\), the law of one price leads us to the conclusion that $V_0=\mathbb{E}(BS(T,X_T,k,v_T))$. Then, we apply the anticipating It\^o's formula of Theorem \ref{aito} in Appendix A to the function $BS(t,X_t,k,v_t)$, observing that $v_t=\sqrt{\frac{Y_t}{T-t}}$ with $Y_t=\int_t^T \sigma^2_s ds$ for $t<T$. That is, $F(t,X_t,Y_t)=BS(t,X_t,k,\sqrt{\frac{Y_t}{T-t}})$.  Then, we get that 
\begin{align*}
\begin{split}
BS(T,X_T,k,v_T) &= BS(0,X_0,k,v_0)+\int_0^T \partial_s BS(s,X_s,k,v_s)ds\\
  & + \int_0^T \partial_x BS(s,X_s,k,v_s)\left(-\frac{1}{2}\sigma_s^2ds + \sigma_sdW_s\right)\\
  & + \int_0^T \partial_{\sigma}BS(s,X_s,k,v_s)\left(\frac{v_s^2}{2(T-s)v_s}-\frac{\sigma_s^2}{2(T-s)v_s}\right)ds  \\
  & + \int_0^T \partial^2_{\sigma x}BS(s,X_s,k,v_s)\frac{\sigma_s}{2(T-s)v_s}\left(\int_s^T D_s^W\sigma_r^2dr\right)ds\\
  & +\frac{1}{2}\int_0^T\partial_{xx}^2 BS(s,X_s,k,v_s) \sigma_s^2 ds.
\end{split}
\end{align*}
By adding and subtracting
\(\frac{1}{2}\int_0^T v_s^2(\partial_{xx}^2BS(s,X_s,k,v_s)-\partial_xBS(s,X_s,k,v_s))ds\)
to the expression above we get that 
\begin{align*}
\begin{split}
&BS(T,X_T,k,v_T) = BS(0,X_0,k,v_0)\\
&\quad +\int_0^T \left(\partial_s BS(s,X_s,k,v_s)-\frac{1}{2}v_s^2 \partial_xBS(s,X_s,k,v_s) + \frac{1}{2}v_s^2 \partial^2_{xx}BS(s,X_s,k,v_s) \right)ds\\
  &\quad + \int_0^T \partial_x BS(s,X_s,k,v_s)\sigma_sdW_s- \int_0^T \partial_{\sigma}BS(s,X_s,k,v_s)\frac{\sigma_s^2-v_s^2}{2(T-s)v_s}ds  \\
  &\quad + \int_0^T \partial^2_{\sigma x}BS(s,X_s,k,v_s)\frac{\sigma_s}{2(T-s)v_s}\left(\int_s^T D_s^W\sigma_r^2dr\right) ds \\
  & \quad +\frac{1}{2}\int_0^T (\partial_{xx}^2BS(s,X_s,k,v_s)-\partial_xBS(s,X_s,k,v_s))(\sigma_s^2-v_s^2)ds.
\end{split}
\end{align*}
Notice that the second term in the above expression is equal to zero due to formula \eqref{bspde}. Finally, using equation \eqref{gvd} and taking expectation, we complete the proof. Observe that by Lemma \ref{bound} and Hypotheses 1-4 all expectations are finite. In fact, using Cauchy-Schwarz inequality, and Hypotheses 1 and (\ref{d1}), we get that
\begin{equation*}
\begin{split}
&\bigg\vert\mathbb{E}\left(\int_0^T H(s,X_s,k,v_s)\sigma_s \left(\int_s^T D_s^W\sigma_r^2dr\right) ds\right)\bigg\vert\\
& \quad \leq C\int_0^T\left(\E(\left|H(s,X_s,k,v_{s})\right|^2)\right)^{1/2} (T-s)^{1/2} \left(\int_s^T \E \left( \vert D_s^W\sigma_r \vert^2\right) dr \right)^{1/2} ds\\
& \quad \leq C \int_0^T (T-s)^{-1} (T-s)^{1/2} (T-s)^H ds=CT^{H+\frac12}.
\end{split}
\end{equation*} 
We also observe that since the function $BS$ and its derivatives are
not bounded, exactly the same truncation argument of Theorem 4.2 in \cite{Alos2007b} can be used here in order to apply Theorem \ref{aito} in Appendix A.
\end{proof}

We next derive an expression for the ATMIV skew of an Inverse European call option under the stochastic volatility model (\ref{B_Epm}). The proof follows similarly to Theorem 4.2 in Alòs {\it{et al.}} \cite{Alos2007b}, but with the use of properties (\ref{p3}) and (\ref{derb}) of the Black-Scholes function for Inverse options, which differ from the standard Black-Scholes function.
\begin{proposition}  \label{skew}
Assume Hypotheses 1-4. Then, 
\begin{align}
\begin{split}
&\lim_{T\to 0} T^{\max(\frac12-H, 0)}\partial_kI(0,k^*) \\
&\qquad = \lim_{T\to 0}T^{\max(\frac12-H, 0)}\frac{\mathbb{E}\left(\int_0^T ( \partial_k H(s,X_s,k^*,v_s) - H(s,X_s,k^*,v_s))\Lambda_s ds\right)}{\partial_{\sigma}BS(0,X_0,k^*,I(0,k^{*}))},
\label{PropSkew}
\end{split}
\end{align}
where $\Lambda_s=\sigma_s\int_s^T D_s^W\sigma_r^2dr$, provided that both limits exist.
\end{proposition}

\begin{proof} Since \(V_0=BS(0,X_0,k^*,I(0,k^*))\), the following equation
holds \begin{equation*}
\partial_kV_0=\partial_kBS(0,X_0,k^*,I(0,k^*))+\partial_{\sigma}BS(0,X_0,k^*,I(0,k^*))\partial_kI(0,k^*).
\end{equation*}
On the other hand, using Theorem \ref{decomposition}, we 
get that \begin{equation*}
\partial_kV_0=\partial_k\mathbb{E}\left(BS(0,X_0,k^*,v_0)\right)+\mathbb{E}\left(\int_0^T \partial_k H(s,X_s,k^*,v_s)\Lambda_s ds\right).
\end{equation*}
Notice that by dominated convergence, we have that
$$
\partial_k\mathbb{E}\left(BS(0,X_0,k^*,v_0)\right)=\mathbb{E}\left(\partial_kBS(0,X_0,k^*,v_0)\right),
$$
where
\begin{equation} \label{p3}
\partial_kBS(0,x,k^*,\sigma)= BS(0,x,k^*,\sigma) - \frac{1}{2}\text{Erfc}\left(\frac{\sqrt{T} \sigma}{2\sqrt{2}}\right).
\end{equation}
Combining the above equations, we find that the volatility skew \(\partial_kI(0,k^*)\) is equal to \begin{equation*}
\begin{split}
&(\partial_{\sigma}BS(0,X_0,k^*,I(0,k^*)))^{-1}\bigg(\mathbb{E}\left(\int_0^T \partial_k H(s,X_s,k^*,v_s)\Lambda_s ds\right)+\mathbb{E}\left(\partial_kBS(0,X_0,k^*,v_0)\right)\\
&\qquad -\partial_kBS(0,X_0,k^*,I(0,k^*))\bigg).
\end{split}
\end{equation*}
Furthermore, using (\ref{p3})
and Theorem \ref{decomposition} we obtain that
\begin{equation*}
\begin{split}
&\mathbb{E}\left(\partial_k BS(0,X_0,k^*,v_0)\right) - \partial_kBS(0,X_0,k^*,I(0,k^*)) \\
&=\left( \mathbb{E}( BS(0,X_0,k^*,v_0)) - V_0 \right)+\frac{1}{2}\left(\text{Erfc}\left(\frac{\sqrt{T} I(0,k^*)}{2\sqrt{2}}\right) - \text{Erfc}\left(\frac{\sqrt{T} v_0}{2\sqrt{2}}\right)\right)\\
&=-\mathbb{E}\left(\int_0^T  H(s,X_s,k^\ast,v_s)\Lambda_sds\right) +\frac{1}{2}\left(\text{Erfc}\left(\frac{\sqrt{T} I(0,k^*)}{2\sqrt{2}}\right) - \text{Erfc}\left(\frac{\sqrt{T} v_0}{2\sqrt{2}}\right)\right).
\end{split}
\end{equation*}

Straightforward differentiation gives us the following expression
\begin{equation} \label{derb}
\begin{split}
\partial_{\sigma}BS(0,X_0,k^*,\sigma)=-\sigma T e^{\sigma ^2T} \text{Erfc}\left(\frac{3 \sigma  \sqrt{T}}{2 \sqrt{2}}\right)+\frac{ e^{-\frac{1}{8} \sigma ^2
   T}\sqrt{T}}{\sqrt{2 \pi }}.
\end{split}
\end{equation}

By (\ref{main1}), $\lim_{T\to 0}I(0,k^{\ast})= \sigma_0$. Moreover, by continuity, we have that $\lim_{T \rightarrow 0} v_0=\sigma_0$.
Thus, $\lim_{T\to 0}I (0,k^{\ast})= v_0$. Thus, since $\text{Erfc}(\sqrt{T}z)=\frac{2\sqrt{T}}{\pi} \int_z^{\infty} e^{-t^2T} dt$, we conclude that
$$
\lim_{T \rightarrow 0} T^{\max(\frac12-H, 0)} \frac{\text{Erfc}\left(\frac{\sqrt{T} I(0,k^*)}{2\sqrt{2}}\right) - \text{Erfc}\left(\frac{\sqrt{T} v_0}{2\sqrt{2}}\right)}{\partial_{\sigma}BS(0,X_0,k^*,I(0, k^{\ast}))}=0,
$$
which completes the proof.
\end{proof}

To compute the limit of the skew slope of the ATMIV, we need to identify the leading-order terms of the numerator in equation \eqref{PropSkew}. The following decomposition formula will be crucial in achieving this goal. The proof follows similarly to that in Alòs {\it{et al.}} \cite{Alos2007b}. For completeness, we provide a sketch of the proof.
\begin{proposition}\label{skew2}
Assume Hypotheses 1-4.  Then, 
\begin{align*}
\begin{split}
\mathbb{E}\left(\int_0^T G(s,X_s,k,v_s)\Lambda_s ds\right)&=\mathbb{E}\left( G(0,X_0,k,v_0) J_0\right)\\
&+\mathbb{E}\left(\frac{1}{2}\int_0^T (\partial_{xxx}^3-\partial_{xx}^2)G(s,X_s,k,v_s) J_s\Lambda_sds\right)\\
&+ \mathbb{E}\left(\int_0^T \partial_{x}G(s,X_s,k,v_s)\sigma_s D^{-}J_sds\right),
\end{split}
\end{align*}
where $J_s=\int_s^T \Lambda_r dr$ and $D^{-}J_s =\int_s^T D_s^W \Lambda_r dr$.
\end{proposition}

\begin{proof}  Applying Theorem \ref{aito} to the function
\((\partial_k H(0,X_0,k,v_0) - H(0,X_0,k,v_0))\int_0^T \Lambda_s ds\) and recalling that $G(s,x,k,v_s)=\partial_k H(s,x,k,v_s) - H(s,x,k,v_s)$, we obtain  
\begin{align*}
&\int_0^TG(s,X_s,k,v_s)\Lambda_sds = G(0,X_0,k,v_0)J_0 \\
  &\qquad +\int_0^T \left(\partial_sG(s,X_s,k,v_s)+\frac{v_s^2}{2(T-s)v_s}\partial_vG(s,X_s,k,v_s)\right)J_sds\\
  &\qquad  + \int_0^T \partial_x G(s,X_s,k,v_s)J_s
  \left(-\frac{1}{2}\sigma_s^2ds + \sigma_sdW_s \right)\\
  &\qquad   - \int_0^T \partial_vG(s,X_s,k,v_s)J_s\frac{\sigma_s^2}{2(T-s)v_s}ds   + \int_0^T \partial_{vx}^2G(s,X_s,k,v_s)J_s\Lambda_s\frac{1}{2(T-s)v_s}ds\\
  &\qquad  + \int_0^T \partial_{x}G(s,X_s,k,v_s)\sigma_s D^{-}J_sds+\frac{1}{2}\int_0^T \sigma_s^2\partial_{xx}^2 G(s,X_s,k,v_s)J_sds.
\end{align*}
By adding and subtracting the term
\(\frac{1}{2}\int_0^T v_s^2(\partial_{xx}^2G(s,X_s,k,v_s)-\partial_xG(s,X_s,k,v_s))ds\)
to the expression above we get that \begin{align*}
\begin{split}
&\int_0^TG(s,X_s,k,v_s)\Lambda_sds = G(0,X_0,k,v_0)J_0 \\
  &\qquad +\int_0^T (\partial_sG(s,X_s,k,v_s)+ \frac{1}{2}v_s^2(\partial_{xx}^2 G(s,X_s,k,v_s) - \partial_x G(s,X_s,k,v_s)) )J_sds\\
  &\qquad + \int_0^T \frac{1}{2}(\partial_{xx}^2 G(s,X_s,k,v_s) - \partial_x G(s,X_s,k,v_s))(\sigma_s^2-v_s^2)J_sds\\
  &\qquad -\int_0^T \partial_vG(s,X_s,k,v_s)\frac{\sigma_s^2-v_s^2}{2(T-s)v_s}J_sds+ \int_0^T \partial_x G(s,X_s,k,v_s)J_s\sigma_sdW_s\\
  &\qquad + \int_0^T \partial_{vx}^2G(s,X_s,k,v_s)J_s\Lambda_s\frac{1}{2(T-s)v_s}ds+ \int_0^T \partial_{x}G(s,X_s,k,v_s)\sigma_s D^{-}J_sds.
  \end{split}
\end{align*}

Next, equations \eqref{bspde} and (\ref{gvd}) imply that
\begin{align*}
\begin{split}
&\partial_s G(s,X_s,k,v_s)-\frac{1}{2}v_s^2 \partial_x G(s,X_s,k,v_s) + \frac{1}{2}v_s^2 \partial^2_{xx} G(s,X_s,k,v_s)=0, \\
&\partial_{xx}^2G(s,X_s,k,v_s)-\partial_xG(s,X_s,k,v_s)=\frac{\partial_{v}G(s,X_s,k,v_s)}{v_s(T-s)}.
\end{split}
\end{align*}
Finally, taking expectations and noticing that by Lemma \ref{bound} and Hypotheses 1-4 all expectations are finite, we complete the desired proof. Remark that as for Theorem \ref{decomposition}, the same truncation argument as in \cite{Alos2007b} can be used in order to apply Theorem \ref{aito} in Appendix A.
\end{proof}

\section{Proof of Theorem \ref{limskew}} 

\subsection{Proof of (\ref{main1}) in Theorem \ref{limskew}: ATMIV
level}

This section is devoted to the proof of (\ref{main1}) in Theorem \ref{limskew}. The proof follows similar ideas to those in Alòs and Shiraya \cite{Alos2019a}, but a detailed study of the inverse of the Black-Scholes function for Inverse options is required, which is provided in Appendix B.

\subsubsection{The uncorrelated case}

Notice that if $\rho=0$, Theorem \ref{decomposition} implies $V_{0}=\mathbb{E}\left( BS( 0,X_0,k^*,v_{0}) \right)$.
Then the implied volatility satisfies the following 
\begin{equation*}
\begin{split}
I^0( 0,k^*) &= BS^{-1}( k^*,V_{0})  = \mathbb{E}\left( BS^{-1}(k^*,\mathbb{E} BS( 0,X_0,k^*,v_{0}) ) \right)  \\
&=\mathbb{E} \left( BS^{-1}(k^*,\Phi_0) -BS^{-1}(k^*,\Phi_T ) \right) +\mathbb{E}\left(v_{0}\right),
\end{split}
\end{equation*}
where $\Phi _{r}:=\E\left( BS\left( 0,X_0,k^*,v_{0}\right) | \mathcal{F}_r^{W}\right)$.

We observe that as $\rho=0$, the two Brownian motions $W$ and $W'$ are independent. Thus, $\Phi_r = \E\left( BS\left( 0,X_0,k^*,v_{0}\right)| \mathcal{F}_r^{W'} \right)$ and $(\Phi_r)_{r \geq 0}$ is a martingale under the probability measure $\mathbb{P}$ with respect to filtration $(\mathcal{F}^{W'}_r)_{r \geq 0}$. By the martingale representation  theorem, there exists a square integrable and  $\mathcal{F}^{W'}$-adapted  process $(U_r)_{r\geq 0}$ such  that 
$$\Phi_r=\Phi_0+\int_0^r U_s dW_s'.$$ 
Clark-Ocone-Haussman formula (Theorem \ref{clar}) gives us the following representation, 
\begin{align}
\begin{split}
U_{r} &=\mathbb{E}\left( D_{r}^{W'}BS( 0,X_0,k^*,v_{0})| \mathcal{F}_r^{W'} \right)  = \mathbb{E}\left( \frac{\partial BS}{\partial \sigma}( 0,X_0,k^*,v_{0}) )D_{r}^{W'}v_0 | \mathcal{F}_r^{W'}\right)  \\
&=\mathbb{E}\left(\frac{\partial BS}{\partial \sigma}( 0,X_0,k^*,v_{0})\frac{\int_{r}^{T}D_{r}^{W'}\sigma _{s}^{2}ds}{2v_{0}}| \mathcal{F}_r^{W'}\right),
\nonumber
\end{split}
\end{align}
Then, a direct application of the classical It\^o's formula implies that
\begin{align*}
\begin{split}
\E\left( BS^{-1}(k^*,\Phi_0 )-BS^{-1}(k^*,\Phi_T )\right) &=-\E\left(\int_{0}^{T}( BS^{-1}) ^{\prime }( k^*, \Phi_r ) U_{r}dW_{r}'\right) \\
& -\E\left(\frac{1}{2}\int_{0}^{T}( BS^{-1}) ^{\prime \prime}( k^*, \Phi_r) U_{r}^{2}dr\right) \\
&=-\E\left(\frac{1}{2}\int_{0}^{T}( BS^{-1}) ^{\prime \prime}( k^*, \Phi_r) U_{r}^{2}dr\right),
\end{split}
\end{align*}
where $(BS^{{-1}})^{\prime}$ and $(BS^{{-1}})^{\prime \prime }$ denote, respectively, the first and  second  derivatives of $BS^{-1}$ with respect to $\sigma$.

Thus, since $T<1$ and by Hypothesis \ref{Hyp1} we get $\left|\partial_{\sigma}BS(0,X_0,k^*,v_0)\right| \leq C \sqrt{T}$.
By Lemma \ref{bound2} (see Appendix \ref{Greeks}) for $T$ sufficiently small, $\left|( BS^{-1}) ^{\prime \prime}( k^*, \Phi_r)\right| \leq C T^{-\frac{1}{2}}$.

Therefore, by Hypotheses \ref{Hyp1} and \ref{Hyp4}  and Cauchy-Schwarz inequality, we get that
\begin{align*}
\begin{split}
\left|\E\left(\int_{0}^{T}( BS^{-1}) ^{\prime \prime}( k^*, \Phi_r) U_{r}^{2}dr\right)\right| &\leq C T^{-1/2}\int_{0}^{T}\E (U_{r}^{2}) \, dr \\
& \leq C T^{-1/2}\int_{0}^{T} (T-r)^{2H+1}ds dr=C T^{2H+\frac32}.
\end{split}
\end{align*}
Thus, $\lim_{T \to 0}\E\left(\int_{0}^{T}( BS^{-1}) ^{\prime \prime}( k^*, \Phi_r) U_{r}^{2}dr\right)=0$. Finally, by continuity, we conclude that
\begin{equation} \label{io}
\lim_{T \to 0} I^0( 0,k^*) = \lim_{T \to 0} \E(v_0) = \sigma_0,
\end{equation}
which completes the proof of (\ref{main1}) in the uncorrelated case.

\subsubsection{The correlated case}

Using similar ideas as in the uncorrelated case we get that
\begin{equation*}
\begin{split} 
I( 0,k^*) &= BS^{-1}( k^*,V_{0})  =  \E\left( BS^{-1}(k^*,\Gamma_T) + BS^{-1}(k^*,\Gamma_0) - BS^{-1}(k^*,\Gamma_0)\right)   \\
&=\E \left(BS^{-1}(k^*,\Gamma_T ) -BS^{-1}(k^*,\Gamma_0 )\right) +I^0( 0,k^*)  \\
&=\E \left( BS^{-1}(k^*,\Gamma_T) -BS^{-1}(k^*,\Gamma_0 ) \right) + I^0( 0,k^*),
\end{split}
\end{equation*}
where $\Gamma _{s}:=\E\left(BS(0,X_0,k^*,v_{0})\right) +\frac{\rho }{2}\E\left(\int_{0}^{s} H(r,X_{r},k^*,v_{r})\Lambda _{r}dr\right)$.

Then, a direct application of It\^o's formula gives us
\begin{equation*}
I( 0,k^*) = I^0( 0,k^*)  + \E \left(\int_{0}^{T}( BS^{-1}) ^{\prime }( k^{\ast},\Gamma _{s}) H(s,X_s,k^*,v_{s})\Lambda _{s}ds\right) .
\end{equation*}
By (\ref{o1}), we have that $(\E(\left|H(s,X_s,k^*,v_{s})\right|^p))^{1/p} \leq C(T-s)^{-1}$. Moreover, using the Lemma \ref{bound3} in  Appendix \ref{Greeks} we have that for $T$ sufficiently small $\left|(BS^{-1}) ^{\prime }( k^{\ast},\Gamma _{s})\right| \leq C(T-s)^{-\frac{1}{2}}$. Therefore, using Hypotheses 1 and \ref{Hyp4} and Cauchy-Schwarz inequality, we get that 
\begin{equation*} 
\begin{split}
&\bigg\vert\E \left(\int_{0}^{T}( BS^{-1}) ^{\prime }( k^{\ast},\Gamma _{s}) H(s,X_s,k^*,v_{s})\Lambda _{s}ds\right) \bigg\vert \\
&\quad \leq C \int_0^T (T-s)^{-1/2} (T-s)^{-1} (T-s)^{H+1/2} ds=C T^H.
\end{split}
\end{equation*}
Thus, $\lim_{T \to 0}\E \left(\int_{0}^{T}( BS^{-1}) ^{\prime }( k^{\ast},\Gamma _{s}) H(s,X_s,k^*,v_{s})\Lambda _{s}ds\right)=0$. Finally, using (\ref{io}) we conclude the proof of (\ref{main1}) in the correlated case.

\subsection{Proof of (\ref{main2}) in Theorem \ref{limskew}: ATM implied volatility
skew}

Appealing to Propositions \ref{skew} and \ref{skew2}
we have that 
\begin{align}
\begin{split}
&\lim_{T\to 0}T^{\max(\frac12-H, 0)}\partial_kI(0,k^*) \\
&= \lim_{T\to 0}T^{\max(\frac12-H, 0)}\frac{1}{\partial_{\sigma}BS(0,X_0,k^*,I(0,k^{*}))}\bigg(\mathbb{E}\left( G(0,X_0,k^*,v_0) J_0\right) \\
&\qquad \qquad +\mathbb{E}\left(\frac{1}{2}\int_0^T (\partial_{xxx}^3-\partial_{xx}^2)G(s,X_s,k^*,v_s) J_s\Lambda_sds\right)\\
&\qquad \qquad + \mathbb{E}\left(\int_0^T \partial_{x}G(s,X_s,k^*,v_s)\sigma_s D^{-}J_sds\right)\bigg).
\end{split}
\label{interm32}
\end{align}

We start by analysing the second term in (\ref{interm32}).
 Using Hypotheses \ref{Hyp1} and \ref{Hyp4}, (\ref{o3}) and Cauchy-Schwarz inequality, we get that
 \begin{equation*}
\begin{split}
&\bigg\vert\mathbb{E}\left(\int_0^T (\partial_{xxx}^3-\partial_{xx}^2)G(s,X_s,k^*,v_s) J_s\Lambda_s ds\right)  \bigg\vert\\
&\leq C \int_0^T (T-s)^{-3} \E\left| J_s\Lambda_s\right|ds \\
&\leq C \int_0^T (T-s)^{-3}\sqrt{\mathbb{E}\left(\left(\int_s^T \left|D_s^{W'}\sigma_r\right|dr\right)^2\right)\mathbb{E}\left(\left(\int_s^T\int_u^T\left|D_u^{W'}\sigma_r\right| dr du\right)^2 \right)} ds \\
&\leq C \int_0^T (T-s)^{-3} (T-s)^{1/2} (T-s)^H (T-s)^{1/2} (T-s)^{H+1} ds =CT^{2H}.
\end{split}
\end{equation*}

Next, we treat the third term in (\ref{interm32}).
We have that
\begin{align*}
\begin{split}
\int_s^TD_s^W\Lambda_rdr &= \int_s^TD_s^W\left(\sigma_r\int_r^TD_r^W\sigma_u^2du\right)dr\\
&= \int_s^T\left(\left(D_s^W\sigma_r\right)\int_r^TD_r^W\sigma_u^2du+\sigma_r\int_r^TD_s^WD_r^W\sigma_u^2du\right)dr,
\end{split}
\end{align*}
where
\begin{align*}
\begin{split}
D_s^WD_r^W\sigma_u^2&=2(D_s^W\sigma_uD_r^W\sigma_u+\sigma_uD_s^WD_r^W\sigma_u).
\end{split}
\end{align*}

Hypothesis \ref{Hyp1} implies that
$$\left|D_s^WD_r^W\sigma_u^2\right| \leq C \left| D_s^{W'}\sigma_u D_r^{W'}\sigma_u + D_s^{W'}D_r^{W'}\sigma_u\right|.$$

Next, Hypotheses 1 and \ref{Hyp4} together with Cauchy-Schwarz inequality yield to 
\begin{equation*}
\begin{split}
\mathbb{E}\left(\sigma_r\int_r^T\left|D_s^WD_r^W\sigma_u^2\right|du\right) & \leq C \int_r^T \left( (u-r)^{H-\frac{1}{2}}(u-s)^{H-\frac{1}{2}}\right)du \\
&\leq C (T-s)^{2H+1}, \\
\mathbb{E}\left(\left|D_s^W\sigma_r\int_r^TD_r^W\sigma_u^2du\right|\right) 
& \leq C \sqrt{\mathbb{E}((D_s^{W'}\sigma_r)^2) \mathbb{E}\left(\left(\int_r^T\left|D_r^{W'}\sigma_u\right|du\right)^2 \right)} \\
&\leq C (r-s)^{H-\frac{1}{2}}(T-r)^{H+\frac{1}{2}}.
\end{split}
\end{equation*}

Then, using the computations above together with (\ref{o2}), we get that
\begin{equation*}
\begin{split}
&\bigg\vert \mathbb{E}\left(\int_0^T \partial_{x}G(s,X_s,k^*,v_s)\sigma_s D^{-}J_s ds\right)  \bigg\vert \\
& \leq \int_0^T (T-s)^{-2} \E \vert D^{-}J_s \vert ds \\
& \leq C \int_0^T (T-s)^{-2} \left( (T-s)^{2H+1} + (T-s)^{2H+2}) \right) ds\\
&\leq C T^{2H}.
\end{split}
\end{equation*}

Finally, using the expression (\ref{derb}), we conclude that the limits of the two terms on the right-hand side of (\ref{interm32}) are zero. Therefore, appealing again to (\ref{derb}), we conclude that
\begin{equation*}
\begin{split}
&\lim_{T \to 0} T^{\max(\frac12-H, 0)}\partial_kI(0,k^*)\\
& = \lim_{T \to 0} T^{\max(\frac12-H, 0)}\frac{\mathbb{E}\left( G(0,X_0,k^*,v_0) J_0\right)}{-I(0,k^{\ast}) T e^{I(0,k^{\ast})^2 T} \text{Erfc}\left(\frac{3 I(0,k^{\ast})  \sqrt{T}}{2 \sqrt{2}}\right)+\frac{ e^{-\frac{1}{8} I(0,k^{\ast})^2
   T}\sqrt{T}}{\sqrt{2 \pi }}}.
\end{split}
\end{equation*}
where
\begin{equation*}
G(0,X_0,k^*,v_0)= \frac{e^{-\frac{1}{8} T v_0^2} \left(T v_0^2+4\right)}{8 \sqrt{2 \pi } T^{3/2} v_0^3}
\end{equation*}
and
$$
J_0=\int_0^T \sigma_s \int_s^T D^W_s \sigma_r^2 dr ds.
$$

Notice that 
\begin{align*}
\mathbb{E} J_0 =  \mathbb{E}\int_0^T 2\rho\left( (\sigma_0+(\sigma_s-\sigma_0))\int_s^T(\sigma_0+(\sigma_r-\sigma_0))D_s^{W'}\sigma_rdr \right)ds.
\end{align*}
Hence, using Hypotheses 1-4 and Cauchy-Schwarz inequality, we get that
\begin{equation*}
\begin{split}
\left|\mathbb{E} J_0-2 \rho \sigma_0^2\int_0^T \int_s^T \mathbb{E}\left(D_s^{W'}\sigma_r\right)dr ds\right| \leq C T^{\frac{3}{2}+H+\gamma}.
\end{split}
\end{equation*}
Therefore, appealing to the Lebesgue dominated convergence theorem, we conclude the proof of (\ref{main2}).

\section{Numerical analysis}

In this section, we justify Theorem \ref{limskew} with numerical simulations. Note that the SABR and fractional Bergomi models do not satisfy Hypothesis \ref{Hyp1}. However, a truncation argument justifies the application of Theorem \ref{limskew}, similar to the approach in Alòs and Shiraya \cite{Alos2019a} and Alòs{\it{et al.}} \cite{Alos2022}. See Appendix C for the details.

\subsection{The SABR model}

We consider the SABR stochastic volatility model (see Hagan{\it{et al.}} \cite{H02}) with a skewness parameter equal to 1, which is the most common case
from a practical point of view. This corresponds to equation (\ref{B_Epm}), where $S_t$ denotes the  forward price of the underlying asset and
$$
    d\sigma_t  = \alpha\sigma_t dW_t', \qquad \sigma_t=\sigma_0e^{\alpha W'_t-\frac{\alpha^2}{2}t}.
  $$ 
  where $\alpha>0$ is the volatility of volatility.
 
For \(r\leq t\), we have
$
D_r^{W'}\sigma_t = \alpha \sigma_t
$
and $\mathbb E \left(D_r^{W'}\sigma_t\right) = \alpha \sigma_0$. Therefore, applying Theorem \ref{limskew} we conclude that (see Appendix C)
\begin{align}
\label{skewSABR}
\begin{split}
\lim_{T \to 0} \partial_k I(0,k^{*}) &= \frac{1}{2}\rho \alpha.
 \end{split}
 \end{align}

Observe that from equation \eqref{skewSABR}, we can infer how the parameters of the SABR model influence the value of the ATM implied volatility skew. Firstly, the skew has a finite value, and its behavior depends on \(\rho\), the correlation between the underlying asset and its volatility, and \(\alpha\), the volatility of volatility. Therefore, the sign of the skew is fully determined by the sign of \(\rho\). Furthermore, the magnitude of the skew is directly proportional to the parameters of the SABR model.

We next proceed with some numerical simulations using the following parameters
$$S_0 = 100, \, T=0.001, \, dt=\frac{T}{50}, \, \alpha = 0.3, \,  \, \sigma_0 =(0.1, 0.2, \dots, 1.4).$$
Observe that the maturity chosen for the simulations, \( T = 0.001 \), is equivalent to half a day, which is realistic, as the shortest maturity options on crypto exchanges are typically daily options.

To estimate the price of the Inverse European call option, we use antithetic variates. The price estimate is calculated as follows:
\begin{align}
\begin{split}
\hat{V}_{sabr} &= \frac{ \frac{1}{N}\sum_{i=1}^N V_T^i + \frac{1}{N}\sum_{i=1}^N V_T^{i,A} }{2},
\end{split}
\label{antithetic}
\end{align}
where \( N = 2,000,000 \) and the subscript \( A \) denotes the value of the call option computed on the antithetic trajectory of a Monte Carlo path.

To recover the implied volatility, we use Brent’s method, which combines the bisection, secant, and inverse quadratic interpolation methods, ensuring guaranteed convergence to a root. This method is efficient and does not require derivatives, making it robust for functions that are difficult to differentiate. Striking a balance between reliability and speed is particularly useful for finding roots of continuous functions in challenging situations.

For estimating the skew, we use the following expression, which allows us to avoid the finite difference approximation of the first-order derivative:
\begin{align}
\begin{split}
\partial_k \hat{I}(0,k^{*}) = \frac{-\partial_kBS(0,X_0,k^*,I(0,k^*))-\mathbb{E}\left(e^{k^*-X_T}1_{X_T\geq k^*}\right)}{\partial_{\sigma}BS(0,X_0,k^*,I(0,k^*))}.
\end{split}
\label{skew_estimator}
\end{align}

In Figures \ref{fig2} and \ref{fig22}, we present the results of a Monte Carlo simulation aimed at numerically estimating the skew and the level of the at-the-money implied volatility of the Inverse European call option under the SABR model. We conclude that the numerical results match the theoretical ones.
\begin{figure}[h]
\centering
\begin{tabular}{ccc}
  \includegraphics[width=60mm]{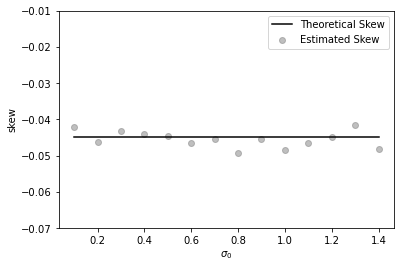} &   \includegraphics[width=60mm]{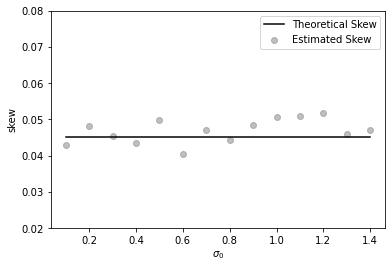} \\
(a) $\rho$=-0.3 & (b) $\rho$=0.3 \\[2pt]
\end{tabular}
\caption{At-the-money skew of the IV under the SABR model.}
\label{fig2}
\end{figure}

\begin{figure}[h]
\centering
\includegraphics[width=60mm]{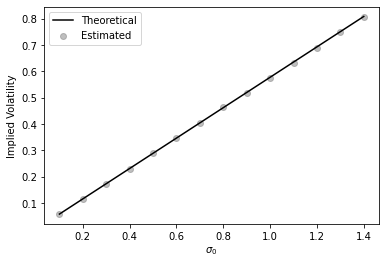}   
\caption{At-the-money level of the IV under the SABR model.}
\label{fig22}
\end{figure}

\subsection{The fractional Bergomi model}

The fractional Bergomi stochastic volatility model (see Alòs {\it{et al.}} \cite[Chapter 5]{Alos2021a})  asssumes equation (\ref{B_Epm}) with
\begin{equation*} \begin{split}
    \sigma_t^2 &= \sigma_0^2 e^{v \sqrt{2H}Z_t-\frac{1}{2}v^2t^{2H}}, \\
    Z_t& = \int_0^t(t-s)^{H-\frac{1}{2}}dW'_s,
    \end{split}
    \end{equation*}
    where $H \in  (0,1)$ and $v>0$.
For $r \leq u$, we have 
\begin{equation*} 
\begin{split}
D_r^{W'}\sigma_u &=  \frac12 \sigma_u v \sqrt{2 H} (u-r)^{H-\frac{1}{2}}, \\
\E(D_r^{W'}\sigma_u) &=e^{-\frac18 v^2 u^{2H}}\frac12  \sigma_0
v \sqrt{2 H} (u-r)^{H-\frac{1}{2}}.
\end{split}
\end{equation*}
Then, applying Theorem \ref{limskew} we obtain that (see Appendix C for the details)
\begin{equation} 
\label{skewRB}
\begin{split}
\lim_{T \to 0} \partial_kI(0,k^{*}) 
 &=\begin{cases}
0 \quad &\text{if} \quad H>\frac12 \\
\frac{\rho v }{4}\quad &\text{if} \quad H=\frac12.
\end{cases}
\end{split}
\end{equation}
and  for $H<\frac12$
\begin{equation} 
\label{skewRB2}
\begin{split}
\lim_{T \to 0} T^{\frac{1}{2}-H}\partial_k I(0,k^{*})&=\frac{2\rho v \sqrt{2 H} }{ \left(3+ 4H(2+H)\right)}.
\end{split}
\end{equation}

In contrast to the SABR model, the fractional Bergomi model introduces the parameter \( H \) into the analytical expression for the skew. Note that equation \eqref{skewRB} relies on condition \( H \geq \frac{1}{2} \). In this case, the ATM implied volatility skew is either zero or depends solely on the product of \( \rho \) and \( v \). Therefore, the sign of the skew is fully determined by the sign of \( \rho \).

We observe that in the case of rough volatility, i.e., \( H < \frac{1}{2} \), the skew behaves as \( T^{H - 1/2} \), which causes it to diverge. However, as shown in (\ref{skewRB2}), a properly scaled skew takes a finite value, and its sign remains determined by the sign of \( \rho \), with the skew being an increasing function of vol-of-vol. The dependence on \( H \) is fully non-linear.

The blow-up of the ATM skew limit in the case of rough volatility is the same phenomenon observed in vanilla options, as discussed in \cite{Alos2007b}. The intuitive explanation is that ATM options are highly uncertain, with a 50/50 chance of becoming in- or out-of-the-money. As the option's maturity shrinks, this uncertainty increases because there is less time for the price to rebound if it moves in favour of the option. This increases the risk of significant losses if unexpected market changes occur. To hedge against this risk, market makers widen spreads in terms of volatility points, which explains the observed IV skew in real market data.
Mathematically, the rough volatility model captures the idea that volatility itself has volatility. Driven by fractional Brownian motion, this model accounts for abrupt shifts in volatility, which are reflected in the adjustment of the IV skew in the Black-Scholes model.

The parameters used for the Monte Carlo simulation are as follows:
$$
S_0 = 100, \, T=0.001, \, dt=\frac{T}{50}, \, H=(0.4, 0.7), \, v = 0.5, \, \rho = -0.3,$$
and $\sigma_0 =(0.1, 0.2, \dots, 1.4)$.

To estimate the price of the Inverse European call option under the fractional Bergomi model, we use antithetic variates as presented in equations \eqref{antithetic}. For the estimation of the skew, we use equation \eqref{skew_estimator}.

In Figure \ref{fig6}, we present the ATM implied volatility skew as a function of maturity for the Inverse European call option, considering two different values of \( H \). We observe the blow-up to \( -\infty \) for the case \( H = 0.4 \). To demonstrate that the fractional Bergomi model captures the power-law structure of the ATM implied volatility skew, we fit the estimated skew to a power-law function \( c T^{-\alpha} \), where \( \alpha \in (0, 0.5) \). Recall that this power-law structure depends on \( H \). According to Theorem 1, when \( H = 0.4 \), the skew should diverge as \( T^{-0.1} \), and when \( H = 0.7 \), it should approach zero as \( T^{0.3} \). However, in Figure \ref{fig6}, we obtain fits of \( T^{-0.059} \) and \( T^{0.237} \), respectively. This difference is attributed to the numerical instability of the finite-difference estimation at short maturities in the presence of rough noise. This issue could be improved by significantly increasing the number of Monte Carlo samples or applying a variance reduction technique.
\begin{figure}[h]
\centering
\begin{tabular}{ccc}
  \includegraphics[width=60mm]{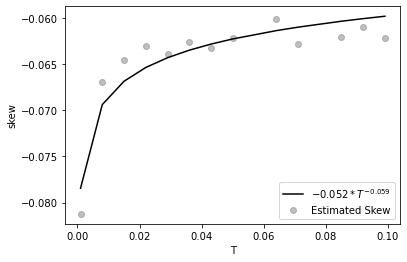} &   \includegraphics[width=60mm]{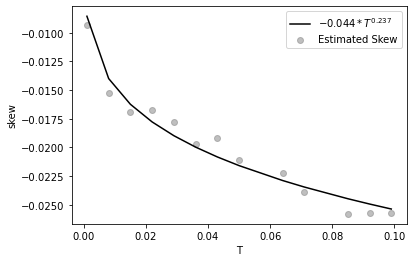} \\
(a) H=0.4, $\sigma_0=0.3$ & (b) H=0.7, $\sigma_0=0.3$ \\[2pt]
\end{tabular}
\caption{At-the-money IV skew as a function of $T$ under fractional Bergomi model}
\label{fig6}
\end{figure}

Due to the blow-up of the at-the-money implied volatility skew of the Inverse European call option when \( H < \frac{1}{2} \), we also plot the quantities \( T^{\frac{1}{2} - H} \partial_k \hat{I}(0, k^{*}) \) for \( H = 0.4 \) and \( \partial_k \hat{I}(0, k^{*}) \) for \( H = 0.7 \) in Figure \ref{fig7}. In Figure \ref{fig77}, we present the estimates of the ATM IV level. We conclude that the theoretical results are in line with the values provided by Theorem \ref{limskew}.
\begin{figure}[h]
\centering
\begin{tabular}{ccc}
  \includegraphics[width=60mm]{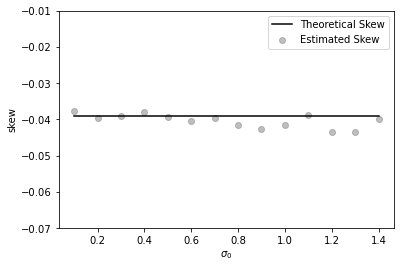} &   \includegraphics[width=60mm]{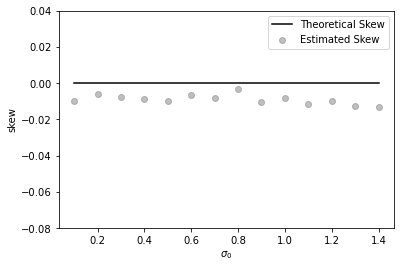} \\
(a) H=0.4 & (b) H=0.7 \\[2pt]
\end{tabular}
\caption{ATMIV skew as a function of $\sigma_0$ under fractional Bergomi model}
\label{fig7}
\end{figure}

\begin{figure}[h]
\centering
\begin{tabular}{ccc}
  \includegraphics[width=60mm]{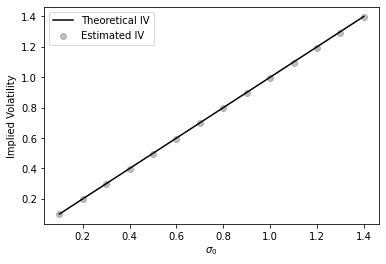} &   \includegraphics[width=60mm]{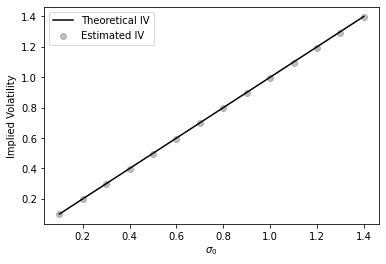} \\
(a) H=0.4 & (b) H=0.7 \\[2pt]
\end{tabular}
\caption{ATMIV level as a function of $\sigma_0$ under fractional Bergomi model.}
\label{fig77}
\end{figure}

\subsection{Empirical Application}

In this section, we demonstrate how the results of Theorem \ref{limskew} can be applied to empirical data. We consider Bitcoin options traded on Deribit\footnote{https://www.deribit.com/options/BTC} on 7 May 2024. The spot price is \( S_0 = 63,500 \) US dollars, and we use 10 different maturities expressed in years, given by the following set of values:
\[
[0.0027, 0.0082, 0.0301, 0.0493, 0.0685, 0.1452, 0.2219, 0.3945, 0.6438, 0.8932].
\]

Since our theoretical results are stated as an asymptotic limit, we plot in Figure \ref{fig99} the implied volatility as a function of the strike for the shortest available maturity, which is 1 day. Recall that the first part of Theorem \ref{limskew} states \( \lim_{T \to 0} I(0, k^*) = \sigma_0 \). Hence, we conclude that the market estimate of \( \sigma_0 \) is approximately 0.36.
\begin{figure}[h]
\centering
  \includegraphics[width=75mm]{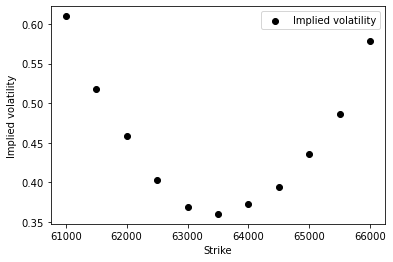} 
\caption{Implied volatility for daily option.}
\label{fig99}
\end{figure}

The second part of Theorem \ref{limskew} says that
\begin{align*}
\begin{split}
\lim_{T \to 0} T^{\max(\frac12-H, 0)}\partial_kI(0,k^{*})&=\lim_{T \to 0} T^{\max(\frac12-H, 0)}\frac{\rho}{ \sigma_0 T^2}\int_0^T\left(\int_r^T\E(D_r^{W'}\sigma_u) du \right) dr.
 \end{split}
\end{align*}

To estimate the implied roughness of the volatility process and identify a class of models that can capture the observed market structure, we begin by estimating the at-the-money implied volatility skew using a standard market approach. 

Note that the crypto derivatives market is sufficiently liquid, allowing us to obtain current market quotes for a wide region of the implied volatility surface. In particular, the most liquid part of the market corresponds to \(0.2-0.5\) delta options, which is exactly the region on which we focus. For a detailed discussion on this topic, see
Mixon \cite{Mixon}.

More specifically, for each available maturity, we set
$$ \partial_k\hat{I}(0,k^{*})= \frac{\Delta_{put}^{0.25}-\Delta_{call}^{0.25}}{\Delta_{call}^{0.5}},$$
Where \( \Delta_{put}^{0.25} \) and \( \Delta_{call}^{0.25} \) represent the market implied volatility of a put and call option with a delta of 0.25, respectively.

Next, we check whether the power-law structure holds for the at-the-money implied volatility skew as a function of maturity. We fit the function \( c \times T^{\alpha} \), where \( \alpha = H - \frac{1}{2} \), to the estimated ATMIV skew using a standard ordinary least squares (OLS) estimator. The result is presented in Figure \ref{fig8}. We observe that the power-law fits the market at-the-money skew adequately, with \( \alpha = 0.3 \), which implies that the market estimation of the Hurst parameter \( H \) from our data is approximately 0.8.

It is important to note that we are not estimating \( H \) from the underlying asset's trajectories but rather from the skew values obtained from the market (see, for instance, Itkin \cite{I24}). Furthermore, the market ATMIV skew for the shortest maturity is 0.014, which is consistent with our theoretical formula (\ref{skewRB}).

Recall that a zero skew means that 0.25 delta calls and puts are priced at the same level of implied volatility, indicating that there is no extra risk premium associated with put options. In contrast to regular markets, where put options are considered insurance against downside movement and are thus in greater demand, leading to elevated implied volatility in comparison to equivalent call options.
\begin{figure}[h]
\centering
  \includegraphics[width=75mm]{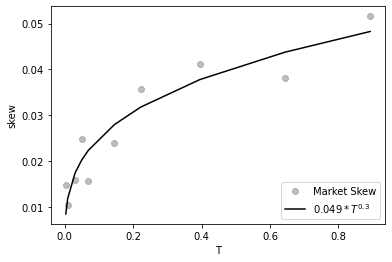} 
\caption{At-the-money implied volatility skew power-law.}
\label{fig8}
\end{figure}

We also observe that the SABR model is not suitable for modeling our observed market implied volatility. There are two main reasons for this. First, it is well-known that the SABR implied volatility surface, along with its approximation through Hagan's formula, cannot accurately reproduce the power-law term structure of the at-the-money skew. Second, formula (\ref{skewSABR}) implies that in order to fit a positive short-end skew, we must force \( \rho \) to be positive. However, the assumption that the underlying spot price is positively correlated with volatility contradicts observed data in regular markets and warrants further investigation.

In summary, through Theorem \ref{limskew}, we are able to identify a class of models that can accurately model the market implied volatility surface and draw meaningful conclusions regarding the market's implied estimates of the corresponding model parameters.

Although the scope of this analysis may seem limited due to the non-monotonicity of implied volatility, this is a common issue in derivatives pricing. For example, calibrating parametric models to the market-implied volatility surface is ill-posed, as the relationship between model parameters and implied volatility is non-unique (e.g., in the SABR, Heston, and fractional Bergomi models). In practice, this non-uniqueness is addressed by selecting the most relevant volatility level for the current market state. Moreover, short-maturity implied volatility often plays a key role in the surface and can serve as a good starting point for approximating the entire curve. While non-uniqueness complicates theoretical analysis, defined values at short maturities often provide reliable approximations for longer maturities.

In terms of trading, short-dated options are actively traded, particularly on platforms like Deribit, which has become a leader in cryptocurrency options. These options, typically expiring daily or weekly, are popular for capitalizing on short-term volatility. Deribit's trading volume data shows a 99$\%$ year-over-year increase in 2024, reflecting the growing demand for options. Traders use short-term options to hedge or exploit event-driven market fluctuations. As a result, there is an increasing need for analytical models for crypto options. The formulas presented here are applicable to options with maturities up to six months, as demonstrated, for example, in Alòs {\it{et al.}}\cite{Alos2022}.

\hypertarget{disclosure-statement}{%
\section*{Disclosure statement}\label{disclosure-statement}}
\addcontentsline{toc}{section}{Disclosure statement}

No potential conflict of interest was reported by the authors.

\appendix

\section{A primer on Malliavin Calculus} \label{MCintro}

We introduce the elementary notions of the
Malliavin calculus used in this paper (see Nualart and Nualart \cite{Nualart2018b}). Consider a standard Brownian motion $Z=\{Z_t\}_{t \in [0,T]}$ defined on a complete probability space $(\Omega, \mathcal{F}, \mathbb{P})$ and the corresponding filtration $\{ \mathcal{F}_t^Z \}_{t \in [0, T]}$ generated by $\{Z_t \}_{t \in [0, T]}$. Let ${\cal S}^Z$ be the set of random variables of the form
\begin{equation*}
F=f(Z(h_{1}),\ldots ,Z(h_{n})),  
\end{equation*}
with $h_{1},\ldots ,h_{n}\in L^2([0,T])$, $Z(h_i)$ denotes the Wiener integral of the function $h_i$, for $i=1,..,n$, and $f\in C_{b}^{\infty }(\mathbb{R}^n) $ 
(i.e., $f$ and all its partial derivatives are bounded). Then the Malliavin 
derivative of $F$, $D^Z F$,  is defined
as the stochastic process given by 
\begin{equation*}
D_{s}^ZF=\sum_{j=1}^{n}{\frac{\partial f}{\partial x_{j}}}(Z(h_{1}),\ldots ,Z(h_{n})) h_j(s), \quad s\in [0,T].
\end{equation*}
This operator is closable from $L^{p}(\Omega )$ to $L^p(\Omega; L^2([0,T]))$, for all $p \geq 1$, and we denote by ${\mathbb{D}}_{Z}^{1,p}$ the
closure of ${\cal S}^Z$ with respect to the norm
$$
||F||_{1,p}=\left( \E\left| F\right|
^{p}+\E||D^Z F||_{L^{2}([0,T])}^{p}\right) ^{1/p}.
$$
We also consider the iterated
derivatives \(D^{Z,n}\) for all integers \(n > 1\) whose domains will be denoted by
\(\mathbb{D}^{n,p}_Z\), for all $p \geq 1$. We will use the notation $\mathbb{L}_Z^{n,p}:=L^p([0,T];{\mathbb{D}}_{Z}^{n,p})$.

We denote by $\delta_Z$ the adjoint of the derivative operator and by Dom $\delta_Z$ it domain. If $u$ belongs to  Dom $\delta_Z$, then $\delta (u)$ is called the Skorohod integral of $u$, since the set of $\{ \mathcal{F}_t^Z \}_{t \in [0, T]}$-adapted processes in $L^2([0,T] \times \Omega)$ is included in Dom $\delta_Z$, and for such processes the Skorohod integral coincides with the It\^o's integral. We shall use the notation $\delta_Z(u)=\int_0^T u_s dZ_s$, for any $u \in $ Dom $\delta_Z$. We recall that $\mathbb{L}_Z^{n,2}$ is included in the domain of $\delta_Z$ for all $n \geq 1$ and that $\E(\delta(u))=0$ for all $u \in $ Dom $\delta_Z$.

One of the main results in Malliavin calculus is the Clark-Ocone-Haussman formula.
\begin{theorem} \label{clar}
Let $F\in \mathbb{D}_{Z}^{1,2}$. Then
$$
F=\E(F)+\int_0^T \E(D^Z_r(F)| \mathcal{F}_r^Z )dZ_r.
$$

\end{theorem}
The following theorem is an extension of classical Ito's lemma for the case of non-anticipating processes, see Al\`os \cite{Alos2006a} for details.
\begin{theorem}[Anticipating It\^o's Formula]
\label{aito}
Consider a process of the form $$X_t=X_0+\int_0^t u_sdZ_s+\int_0^t v_s ds,$$ where $X_0$ is a $\mathcal{F}_0^Z$-measurable random variable and $u$ and $v$ are $\{ \mathcal{F}_t^Z \}_{t \in [0, T]}$-adapted processes in $L^2([0,T] \times \Omega)$.

Consider also a process $Y_t = \int_t^T \theta_sds$, for some $\theta \in \mathbb{L}_Z^{1,2}$. Let $F : [0,T] \times \mathbb{R}^2 \rightarrow \mathbb{R}$ be a function $C^{1,2}([0,T] \times \mathbb{R}^2)$ such that there exists a positive constant $C$ such that, for all $t \in [0,T]$, $F$ and its derivatives evaluated in $(t,X_t,Y_t)$ are bounded by $C$. Then it follows that for all $t \in [0,T]$,
\begin{align*}
F(t,X_t,Y_t) &= F(0,X_0,Y_0)+\int_0^t \partial_sF(s,X_s,Y_s)ds+\int_0^t \partial_xF(s,X_s,Y_s) u_s dZ_s \\
&\qquad +\int_0^t \partial_xF(s,X_s,Y_s) v_s ds+\int_0^t \partial_yF(s,X_s,Y_s)dY_s\\
&\qquad +\int_0^t \partial_{xy}^2F(s,X_s,Y_s)u_sD^{-}Y_s ds+\frac{1}{2}\int_0^t \partial_{xx}^2F(s,X_s,Y_s)u_s^2 ds,
\end{align*}
where $D^{-}Y_s=\int_s^TD^Z_s\theta_rdr$ and the integral $\int_0^t \partial_xF(s,X_s,Y_s) u_s dZ_s$ is a Skorohod integral since the process $\partial_xF(s,X_s,Y_s) u_s$ is not adapted.
\end{theorem}

\section{The inverse of the ATM Inverse call option price} \label{Greeks}

Recall that ATM value of an Inverse call option is given by
\begin{align*}
\begin{split}
BS(0,x,k^*,\sigma) &=\frac{1}{2} \left(\text{Erfc}\left(\frac{\sigma  \sqrt{T}}{2 \sqrt{2}}\right)-e^{\sigma ^2T} \text{Erfc}\left(\frac{3 \sigma 
   \sqrt{T}}{2 \sqrt{2}}\right)\right).
\end{split}
\end{align*}

We plot this function in Figure \ref{fig9} as a function of \( \sigma \sqrt{T} \) and observe that the function is not monotonic over its entire domain. As a result, the inverse with respect to \( \sigma \) will not be uniquely defined. However, since we are primarily interested in the behavior of the function for small values of \( T \), in this region the function is monotonically increasing, and the inverse will be well-defined within a small positive interval around zero.
\begin{figure}[h]
\centering
\includegraphics[width=60mm]{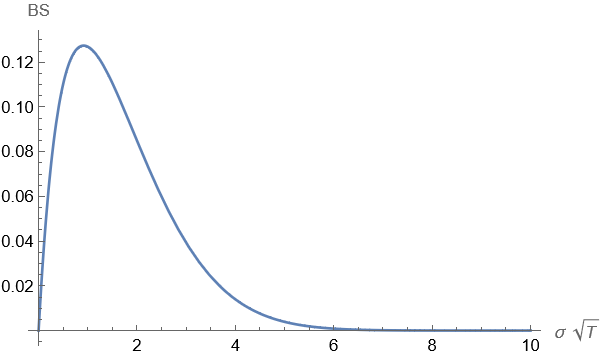} 
\caption{The function $BS(0, x, k^{\ast}, \sigma)$.}
\label{fig9}
\end{figure}

\begin{lemma} \label{bound3}
Under Hypothesis \ref{Hyp1}, for all $\epsilon>0$ sufficiently small there exists a positive constant $C(\epsilon)$
such that for all $0 \leq s <T \leq \epsilon$, 
\begin{equation*}
\begin{split}
\vert ( BS^{-1}) ^{\prime}( k^*, \Gamma_s) \vert \leq  C(\epsilon) T^{-\frac12},
\end{split}
\end{equation*}
where $\Gamma_s=\E\left(BS(0,X_0,k^*,v_{0})\right) +\frac{\rho }{2}\E\left(\int_{0}^{s} H(r,X_{r},k^*,v_{r})\Lambda _{r}dr\right)$.
\end{lemma}

\begin{proof} 
We have that
\begin{align} \label{inversep}
\begin{split}
( BS^{-1}) ^{\prime}( k^*, \Gamma_s) =\frac{1}{\partial_{\sigma}BS(0,X_0,k^*,BS^{-1}( k^*, \Gamma_s)))}.
\end{split}
\end{align}
Recall form (\ref{derb}) that
\begin{align} \label{inversep3}
\begin{split}
\partial_{\sigma}BS(0,X_0,k^*,BS^{-1}( k^*, \Gamma_s))&=\sqrt{T}\left(-y e^{y^2} \text{Erfc}\left(\frac{3y}{2 \sqrt{2}}\right)+\frac{ e^{-\frac{1}{8} y^2}}{\sqrt{2 \pi }}\right)=:g(y)\sqrt{T},
\end{split}
\end{align}
where $y = BS^{-1}( k^*, \Gamma_s)\sqrt{T}$.
Therefore, we conclude that
\begin{align*}
\begin{split}
( BS^{-1}) ^{\prime}( k^*, \Gamma_s) &=\frac{1}{g(y)\sqrt{T}},
\end{split}
\end{align*}
Notice that \( g(0) > 0 \), and for sufficiently small \( x > 0 \), the function \( g(y) \) is non-negative and monotonically decreasing on the interval \( [0, x] \). Since \( g(y) \) is continuous, it is lower bounded by a positive constant on that interval. Additionally, \( \Gamma_s \) is also small for sufficiently small \( T \), ensuring that the inverse is well-defined in this case. This completes the proof.
\end{proof}

\begin{lemma} \label{bound2}
Under Hypothesis \ref{Hyp1}, for all $\epsilon>0$ sufficiently small there exists a positive constant $C(\epsilon)$
such that for all $0 \leq r <T \leq \epsilon$, 
\begin{equation*}
\begin{split}
\vert ( BS^{-1}) ^{\prime \prime}( k^*, \Phi_r) \vert  \leq C(\epsilon) T^{-\frac{1}{2}},
\end{split}
\end{equation*} 
where $\Phi_r = \E\left( BS\left( 0,X_0,k^*,v_{0}\right)| \mathcal{F}_r^{W'} \right)$.
\end{lemma}

\begin{proof} 
We have that
\begin{align*}
\begin{split}
( BS^{-1}) ^{\prime \prime}( k^*_t, \Phi_r) &=-\frac{\partial^2_{\sigma\sigma}BS(0,X_0,k^*,BS^{-1}( k^*, \Phi_r)))}{\left(\partial_{\sigma}BS(0,X_0,k^*,BS^{-1}( k^*, \Phi_r)))\right)^3}.
\end{split}
\end{align*}
Therefore,
\begin{align*}
\begin{split}
( BS^{-1}) ^{\prime \prime}( k^*, \Phi_r) &=\frac{11 \sqrt{2} \pi  e^{\frac{y^2}{4}} y-8 \pi ^{3/2} e^{\frac{11 y^2}{8}} \left(2 y^2+1\right) \text{Erfc}\left(\frac{3 y}{2
   \sqrt{2}}\right)}{\sqrt{T} \left(2 \sqrt{\pi } e^{\frac{9 y^2}{8}} y \text{Erfc}\left(\frac{3 y}{2 \sqrt{2}}\right)-\sqrt{2}\right)^3} =f(y)T^{-\frac{1}{2}},
\end{split}
\end{align*}
where $y = BS^{-1}(k^*,\mathbb{E}_r BS( 0,X_0,k^*,v_{0}) )\sqrt{T}$. 
Notice that \( f(0) > 0 \), and for sufficiently small \( x > 0 \), the function \( f(y) \) is non-negative and monotonically increasing on the interval \( [0, x] \). Since \( f(y) \) is continuous, it is upper bounded by a positive constant on that interval. Additionally, \( \Phi_r \) is small for sufficiently small \( T \), ensuring that the inverse is well-defined in this case. This completes the proof.
\end{proof}

\begin{lemma} \label{bound33}
Under Hypothesis \ref{Hyp1}, for all $\epsilon>0$ sufficiently small there exists a positive constant $C(\epsilon)$
such that for all $0 \leq T \leq \epsilon$, 
\begin{equation*}
\begin{split}
\vert \partial^2_{\sigma k}  ( BS^{-1}) ( k^*, \sigma) \vert  \leq C(\epsilon) T^{-\frac{1}{2}},
\end{split}
\end{equation*} 
\end{lemma}

\begin{proof} 
Appealing to (\ref{inversep}), we have that
\begin{align} \label{inversep2}
\begin{split}
\partial^2_{\sigma k}( BS^{-1}) (k^{\ast}, \sigma) &=-\frac{\partial^2_{\sigma k} BS(0,X_0,k^{\ast},BS^{-1}( k^{\ast}, \sigma))) }{\left(\partial_{\sigma}BS(0,X_0,k^{\ast},BS^{-1}( k^{\ast}, \sigma)))\right)^2},
\end{split}
\end{align}
where
$$
BS(0,X_0,k,\sigma)= \frac{1}{2}\left(\text{Erfc}\left(\frac{2 k+\sigma^2 T-2 X_0}{2 \sqrt{2}\sigma  \sqrt{T}}\right)-e^{k+\sigma ^2 T-X_0}\text{Erfc}\left(\frac{2 k+3 \sigma ^2 T-2 X_0}{2 \sqrt{2}\sigma  \sqrt{T}}\right)\right).
$$
Differentiating, we get that
$$
\partial^2_{\sigma k}BS(0,X_0,k^{\ast},\sigma)=
\frac{3 \sqrt{T}e^{-\frac{\sigma ^2 T}{8}}}{2 \sqrt{2 \pi }}-\sigma  T e^{\sigma ^2 T}\text{Erfc}\left(\frac{3 \sigma  \sqrt{T}}{2 \sqrt{2}}\right).
$$
Therefore, using (\ref{inversep3}) and (\ref{inversep2}), we conclude that
$$
\partial^2_{\sigma k}( BS^{-1}) (k^{\ast}, \sigma) =\frac{3 \sqrt{2 \pi }e^{\frac{y^2}{8}}-4 \pi  y e^{\frac{5  y^2}{4}}\text{Erfc}\left(\frac{3 y}{2 \sqrt{2}}\right)}{\sqrt{T}\left(\sqrt{2}-2 \sqrt{\pi }y e^{\frac{9 y^2}{8}}\text{Erfc}\left(\frac{3 y}{2
   \sqrt{2}}\right)\right)^2}=h(y) T^{-1/2},
   $$
   where $y=BS^{-1} (k^{\ast}, \sigma) \sqrt{T}$. This clearly implies the desired result noticing that $h(0)>0$ and for $x>0$ sufficiently small, $h$ is upper bounded on the interval $[0,x]$, which concludes the proof.
\end{proof}

\section{Truncation argument}

We present here the truncation argument needed in order to apply Theorem \ref{limskew} for the SABR and fractional Bergomi models introduced in Sections 5.1 and 5.2.

We start with the SABR model defined in Section 5.1.
We define $\varphi(x)=\sigma_0 \exp(x)$.  For every  $n>1$, we consider a function $\varphi_n \in C^2_b$  satisfying that $\varphi_n(x)=\varphi(x)$ for any $x \in [-n,n]$, 
$\varphi_n(x) \in [\varphi(-2n) \vee \varphi(x), \varphi(-n)]$ for $x \leq -n$, and $\varphi_n(x) \in  [\varphi(n), \varphi(x) \wedge \varphi(2n)]$
for $x \geq n$. We set
$$
\sigma_t^n=\varphi_n\left(\alpha W'_t-\frac{\alpha^2}{2}t\right).
$$
It is easy to see that  $\sigma_t^n$ satisfies Hypotheses 1, 2, 3, and (\ref{d1}). In  fact, for \(r\leq t\), we have that
$$
D_r^{W'}\sigma^n_t = \varphi_n'\left(\alpha W'_t-\frac{\alpha^2}{2}t\right)\alpha,
$$
which implies that  (\ref{d1}) holds with $H=\frac12$ and Hypothesis 2 is satisfied with $\gamma<1/2$. Therefore, appealing to Theorem \ref{limskew} and using the fact that $\sigma_0^n=\sigma_0$, we conclude that for all $n>1$,
\begin{equation} \label{limit1}
\lim_{T\to 0}I^n(0,k^{*})=\sigma_0.
\end{equation}
where $I^{n}$ denotes the implied volatility under the volatility process $\sigma^n_t$. We then write
\begin{equation} \label{op}
I(0,k^{*})=I^n(0,k^{*})+I(0,k^{*})-I^n(0,k^{*}).
\end{equation}
By the mean value theorem,
\begin{equation*} \begin{split}
I(0,k^{*})-I^n(0,k^{*})&=\partial_\sigma (BS^{-1}(0,X_0,X_0,\xi))(V_0-V_0^n),
\end{split}
\end{equation*}
for some $\xi\in (V_0, V_0^n)$, where $V_0^n$ is the option price under $\sigma^n$ and $BS^{-1}(0,X_0,X_0,\xi)$ is defined in Appendix B.
Thus, for $T$ sufficiently small and $n>\alpha^2$, appealing to Lemma \ref{bound3} we get that
\begin{align*}
&\vert I(0,k^{*})-I^n(0,k^{*}) \vert \leq  \frac{C_n}{\sqrt{T}}\E\left (|e^{-X_T}-e^{-X^n_T}|{\bf1}_{\sup_{s\in [0,T]} \vert \ln (\sigma_s/\sigma_0) \vert >n}\right)\nonumber\\
&\qquad \qquad \leq  \frac{C_n}{\sqrt{T}}\E [(|e^{-X_T}+e^{-X^n_T}|^2)]^{1/2} \left[\P \left(\sup_{s\in [0,T]}\vert \ln (\sigma_s/\sigma_0) \vert>n \right)\right]^{1/2}\nonumber\\
&\qquad \qquad\leq \frac{C_n}{\sqrt{T}}\left[\P\left(\sup_{s\in [0,T]}|\alpha W_s'-\alpha^2s/2|>n\right)\right]^{\frac12}\nonumber \\
&\qquad \qquad\leq \frac{C_n}{\sqrt{T}}\left[\P\left(\sup_{s\in [0,T]}| W_s'|> \frac{\alpha}{2}\right)\right]^{\frac12}\nonumber
\end{align*}
for some constant $C_n>0$ that changes from line to line.
Then, Markov's inequality 
implies that  for all $p>0$,
$$\vert I(0,k^{*})-I^n(0,k^{*})\vert\leq \frac{C_n}{\sqrt{T}} 
\left[\E\left( \sup_{s\in [0,T]}|W_s'|^p\right)\right]^{1/2}\leq C_n T^{\frac{p}{4}-\frac12},
$$
 Thus, taking $p>2$ and using (\ref{limit1}) and (\ref{op}), we conclude that  
\begin{equation} \label{lim} 
 \lim_{T\to 0} I(0,k^{*}) =\sigma_0,
 \end{equation}
 which shows the validity of (\ref{main1}) for the SABR model.
 
Next, we prove (\ref{skewSABR}).
For $s \leq r \leq t$, we have
$$
D_s^{W'} D_r^{W'}\sigma^n_t = \varphi_n^{''}\left(\alpha W'_t-\frac{\alpha^2}{2}t\right)\alpha^2,
$$
which implies that (\ref{d2}) holds with $H=\frac12$.
Therefore, appealing to Theorem \ref{limskew} we get that 
$$
 \lim_{T \to 0} \partial_k I^n(0,k^{*}) = \lim_{T \rightarrow 0}\frac{\rho \alpha}{ \sigma_0^n T^2}\int_0^T\left(\int_r^T\E\left(\varphi'_n\left(\alpha W_u'-\frac{\alpha^2}{2}\right)\right) du \right) dr.
$$
Since $\varphi_n \in C^2_b$, using dominated convergence theorem, we get that uniformly for all $T >0$,
$$
\lim_{n \rightarrow \infty}\frac{\rho \alpha}{ \sigma_0^n T^2}\int_0^T\left(\int_r^T\E\left(\varphi'_n\left(\alpha W_u'-\frac{\alpha^2}{2}\right)\right) du \right) dr=\frac12 \rho \alpha.
$$

Next, similarly as above we can write
\begin{equation*} \begin{split}
\partial_k I(0,k^{*})=\partial_k I^n(0,k^{*})+\partial_k(I(0,k^{*})-I^n(0,k^{*})).
\end{split}
\end{equation*}

By the mean value theorem,
\begin{equation*} \begin{split}
\partial_k(I(0,k^{*})-I^n(0,k^{*}))&=\partial_\sigma \partial_k (BS^{-1}(0,X_0,X_0,\xi))(V_0-V_0^n),
\end{split}
\end{equation*}
for some $\xi\in (V_0, V_0^n)$. 
Then, appealing to Lemma \ref{bound33} and proceeding exactly as above, we get that
for sufficiently large $n$,
$$
\lim_{T \rightarrow 0} \partial_k(I(0,k^{*})-I^n(0,k^{*}))=0,
$$
which concludes the proof of (\ref{skewSABR}).

We next prove (\ref{lim}), (\ref{skewRB}) and (\ref{skewRB2}) for the fractional Bergomi model.
We define $\varphi$ and $\varphi_n$ as for the SABR model, and we set
$$
\sigma_t^n=\varphi_n\left(\frac{1}{2} v \sqrt{2H}Z_t-\frac{1}{4}v^2t^{2H}\right).
$$
It is easy to see that  $\sigma_t^n$ satisfies Hypotheses 1, 2, 3, and (\ref{d1}).  In fact, for \(r\leq t\), we have that
$$
D_r^{W'}\sigma^n_t = \varphi_n'\left(\frac{1}{2} v \sqrt{2H}Z_t-\frac{1}{4}v^2t^{2H}\right)\frac12 
v \sqrt{2 H} (t-r)^{H-\frac{1}{2}},
$$
which implies that Hypothesis (\ref{d1}) holds  and Hypothesis 2 is satisfied with $\gamma<H$.
Moreover, for $s \leq r \leq t$, we have
$$
D_s^{W'} D_r^{W'}\sigma^n_t = \varphi_n^{''}\left(\frac{1}{2} v \sqrt{2H}Z_t-\frac{1}{4}v^2t^{2H}\right)\frac{1}{4} 
v^2 \sqrt{4 H^4} (t-r)^{H-\frac{1}{2}}(t-s)^{H-\frac{1}{2}},
$$
which implies that  (\ref{d2}) holds.
Therefore, by Theorem \ref{limskew} we get that (\ref{limit1}) holds.
We next follow the same computations as the SABR model, but using the fractional Bergomi, to get that for $T$ sufficiently small and $n>v^2$,
\begin{align*}
&\vert I(0,k^{*})-I^n(0,k^{*}) \vert \leq  \frac{C_n}{\sqrt{T}}\E\left (|e^{-X_T}-e^{-X^n_T}|{\bf1}_{\sup_{s\in [0,T]} \vert \ln (\sigma_s/\sigma_0) \vert >n}\right)\nonumber\\
&\qquad \qquad \leq  \frac{C_n}{\sqrt{T}}\E [(|e^{-X_T}+e^{-X^n_T}|^2)]^{1/2} \left[\P \left(\sup_{s\in [0,T]}\vert \ln (\sigma_s/\sigma_0) \vert>n \right)\right]^{1/2}\nonumber\\
&\qquad \qquad\leq \frac{C_n}{\sqrt{T}}\left[\P\left(\sup_{s\in [0,T]}|\alpha \frac{1}{2} v \sqrt{2H}Z_s-\frac{1}{4}v^2 s^{2H}|>n\right)\right]^{\frac12} \\
&\qquad \qquad\leq \frac{C_n}{\sqrt{T}}\left[\P\left(\sup_{s\in [0,T]}|Z_s|>\frac{v}{\alpha}\right)\right]^{\frac12},
\end{align*}
for some constant $C_n>0$.
Then, Markov's inequality 
implies that  for all $p>0$
$$\vert I(0,k^{*})-I^n(0,k^{*})\vert\leq \frac{C_n}{\sqrt{T}}  \left[\E\left( \sup_{s\in [0,T]}|\int_0^s (s-r)^{H-\frac12} dW'_s|^p\right)\right]^{1/2}\leq C_n T^{\frac{p H}{2}-\frac12},
$$
Thus, taking $p>\frac{1}{H}$, the above shows that for $n$ sufficiently large,
$$
\lim_{T \rightarrow 0} (I(0,k^{*})-I^n(0,k^{*}))=0,
$$
which concludes the proof of (\ref{lim}).

Concerning the proof of (\ref{skewRB}) and (\ref{skewRB2}), on one hand, appealing to Theorem \ref{limskew}, we get that 
$$
 \lim_{T \to 0} T^{\max(\frac12-H, 0)} \partial_kI^n(0,k^{*})=\lim_{T \rightarrow 0}T^{\max(\frac12-H, 0)}\frac{\rho \alpha}{ \sigma_0^n T^2}\int_0^T\left(\int_r^T\E\left(D_r^{W'} \sigma_u^n\right) du \right) dr.
$$ 
On the other hand, following along the same lines as above, one can easily show that for $n$ sufficiently large 
\begin{equation} \label{skewlo}
\lim_{T \rightarrow 0} \partial_k(I(0,k^{*})-I^n(0,k^{*}))=0.
\end{equation}
Thus, using dominated convergence we get that uniformly for all $T>0$,
$$
\lim_{n \rightarrow \infty}\frac{\rho \alpha}{ \sigma_0^n T^2}\int_0^T\left(\int_r^T\E\left(D_r^{W'} \sigma_u^n\right) du \right) dr=
\frac{\rho \alpha}{ \sigma_0 T^2}\int_0^T\left(\int_r^T\E\left(D_r^{W'} \sigma_u\right) du \right) dr.
$$
Then, computing this integral and using (\ref{skewlo}), we conclude that
(\ref{skewRB}) and (\ref{skewRB2}) hold true.

\bibliography{literature}
\end{document}